\documentclass[a4paper,UKenglish,cleveref,thm-restate%
]{lipics-v2021}

\newboolean{longversion}
\setboolean{longversion}{true}
\NewDocumentCommand \shortlongversion {mm}
	{\ifthenelse{\boolean{longversion}}{#2}{#1}}

\shortlongversion{}{\hideLIPIcs} 

\usepackage{adjustbox}
\usepackage{alphabeta}
\usepackage{mathtools}
\usepackage{csquotes}
\usepackage{environ}
\usepackage{todonotes}
\usepackage{trimclip}
\usepackage{xspace}
\usepackage{stmaryrd}

\usepackage[inline]{enumitem}
\NewDocumentEnvironment{ienumerate}{}
	{\begin{enumerate*}[label=
		{\color{lipicsGray}\sffamily\bfseries\upshape\mathversion{bold}%
		(\roman*)} ]}
	{\end{enumerate*}}
\NewDocumentEnvironment{thmenumerate}{sm}
	{	\IfBooleanT{#1}{\mbox{}\vspace*{-\topsep}}%
		\begin{enumerate}[ref=\arabic{#2}(\arabic*)] }
	{	\end{enumerate} }

\usepackage{ebproof}
\NewDocumentCommand \proofrule {m} {\ensuremath{(\mathsf{#1})}}
\ebproofset{
	right label template	= {\small\proofrule{\inserttext}},
	center					= false,
	rule separation			= 1.6pt,
}
\usetikzlibrary{decorations.pathmorphing}
\ebproofnewrulestyle{snake}{
	rule code={\hbox{\tikz
	\draw[decorate,decoration={snake,amplitude=.2ex}]
	(0,0) -- (\hsize,0);}}
}
\ebproofnewstyle{doubleprem}{
	template={\tunderline{$\strut\inserttext$}}
}
\ebproofnewstyle{ordoubleprem}{
	template={\tunderdash{$\strut\inserttext$}}
}
\NewEnviron{myprooftree}[1][]{%
	\scalebox{.75}{\begin{prooftree}[#1]\BODY\end{prooftree}}%
}

\usepackage{stackengine}
\stackMath
\setstackgap{S}{1pt}

\makeatletter
\patchcmd{\section}{\bfseries}{\bfseries\mathversion{bold}}{}{}
\patchcmd{\subsection}{\bfseries}{\bfseries\mathversion{bold}}{}{}
\patchcmd{\paragraph}{\bfseries}{\bfseries\mathversion{bold}}{}{}
\patchcmd{\@evenhead}{\bfseries}{\bfseries\mathversion{bold}}{}{}
\patchcmd{\@oddhead}{\bfseries}{\bfseries\mathversion{bold}}{}{}
\makeatother

\NewDocumentCommand{\myparagraph}{m}
	{\vspace{1.5ex plus .2ex}\par\noindent{\sffamily\bfseries 
	#1.}\hspace{.8em plus .2em minus .2em}}

\newtheorem{fact}[theorem]{Fact}
\theoremstyle{definition}
\newtheorem{examples}[theorem]{Examples}

\patchcmd{\descriptionlabel}{\sffamily}{\sffamily\upshape}{}{}

\NewDocumentCommand \commentaire {om}
	{{\color{red}\sffamily
	$\blacksquare$
	\IfValueT{#1}{\textbf{#1:}} 
	#2}\par}

\NewDocumentCommand \defemph {m} {\emph{#1}}

\NewDocumentCommand \ie {} {\emph{i.e.}\@\xspace}
\NewDocumentCommand \eg {} {\emph{e.g.}\@\xspace}
\NewDocumentCommand \wrt {} {wrt.\@\xspace}
\NewDocumentCommand \ifff {} {iff\xspace}

\NewDocumentCommand \co {} {\color{orange}}

\NewDocumentCommand \itrs {} {\textsc{itrs}\xspace}

\AtBeginDocument{
	\let\oldEpsilon\Epsilon
	\renewcommand\Epsilon{\mathrm\oldEpsilon}
	\let\oldZeta\Zeta
	\renewcommand\Zeta{\mathrm\oldZeta}
	\let\oldEta\Eta
	\renewcommand\Eta{\mathrm\oldEta}
}

\NewDocumentCommand \longrightharpoondown {}
	{ \relbar\joinrel\rightharpoondown }

\makeatletter
\NewDocumentCommand \longrightsquigarrow {}
	{\mathrel{\mathpalette\@squig\relax}%
		\joinrel\mathrel{\mathpalette\@squig\relax}%
		\joinrel\mathrel{\mathpalette\@squig\relax}%
		\joinrel\rightsquigarrow}
	\newcommand*{\@squig}[2]{%
	   \clipbox{0pt 0pt {.46\width} {-.2\height}}{$\m@th#1\rightsquigarrow$}%
	}
\makeatother

\newlength{\widerelwidth}

\NewDocumentCommand \wide	{o}
	{	\hspace{.6em plus .3em}
		\IfValueT{#1}{\mathord{#1}\hspace{.6em plus .3em}} }

\NewDocumentCommand \tunderline		{m}
    {\tikz[baseline=(foo.base)]{
        \node[inner xsep=0pt,inner ysep=2pt,outer sep=0pt] (foo) {#1};
        \draw (foo.south west) -- (foo.south east);
    }}
\NewDocumentCommand \tunderdash		{m}
    {\tikz[baseline=(foo.base)]{
        \node[inner xsep=0pt,inner ysep=2pt,outer sep=0pt] (foo) {#1};
        \draw[densely dashed] (foo.south west) -- (foo.south east);
    }}

\NewDocumentCommand \Bool 			{}	
	{ \{0,1\} }
\NewDocumentCommand \Nat 			{}	
	{ \mathbf{N} }

\NewDocumentCommand	\set 			{mo}
	{ \left\{ #1 \IfValueT{#2}{\ \middle|\ #2 } \right\} }

\NewDocumentCommand \partialto		{}		{ \rightharpoonup }

\NewDocumentCommand \property		{m}		{ \mathfrak{#1} }

\NewDocumentCommand \drule 			{m} 	{ \focons{#1} }
\NewDocumentCommand \later			{}		{ \mathrel{\triangleright} }

\NewDocumentCommand \Var 			{}		{ \mathcal{V} }
\NewDocumentCommand \var 			{m} 	{ \mathnormal{#1} }

\NewDocumentCommand \hole 			{} 		{ \ast }

\NewDocumentCommand \arity 			{}		{ \mathrm{ar} }

\NewDocumentCommand \trunc			{omo}
	{ \left\lfloor #2 
	\right\rfloor_{\IfValueTF{#1}{#1}{d}}^{\IfValueT{#3}{#3}} }
\NewDocumentCommand \dist			{o}		
	{ \mathbf{d}^{\IfValueT{#1}{#1}} }

\NewDocumentCommand \focons 		{m} 	{ \mathsf{#1} }

\NewDocumentCommand \foterms		{o}
	{ \mathrm{T}_{\IfValueTF{#1}{#1}{\Sigma}} }
\NewDocumentCommand \foiterms		{o}
	{ \mathrm{T}_{\IfValueTF{#1}{#1}{\Sigma}}^\infty }

\NewDocumentCommand \lterms			{}
	{ \Lambda }

\NewDocumentCommand \liiiterms 		{mmm}
	{ \lterms^{#1#2#3} }

\NewDocumentCommand \subst 			{om}	{[ #2 / \IfValueTF{#1}{#1}{x}] }

\NewDocumentCommand \dtrees		{o}
	{ \mathrm{DT}_{\IfValueTF{#1}{#1}{\mathcal D}}^\infty }

\NewDocumentCommand \coind 			{o}	
	{ \mathrm{coind}_{\IfValueTF{#1}{#1}{\drule r}} }

\NewDocumentCommand \red 			{}		
	{ \longrightarrow }
\NewDocumentCommand \reds 			{}		
	{ \red^* }

\colorlet{depth}{magenta!50}
\NewDocumentCommand \redd			{sm}
	{ \red_{\color{depth} \IfBooleanT{#1}{\geq} #2} }

\NewDocumentCommand \redi			{o}		
	{ \mathrel{\stackunder{\red}{\scriptstyle\IfValueT{#1}{#1}}}^\infty }
\NewDocumentCommand \redig 			{o}
	{ \mathrel{
	\stackunder{\longrightharpoondown}{\scriptstyle\IfValueT{#1}{#1}}
	}^\infty }
\NewDocumentCommand \redis			{o}		
	{ \mathrel{\stackunder{\red}{\scriptstyle\IfValueT{#1}{#1}}}^{\infty*} }
\NewDocumentCommand \redigs			{o}		
	{ \mathrel{
	\stackunder{\longrightharpoondown}{\scriptstyle\IfValueT{#1}{#1}}
	}^{\infty*} }

\NewDocumentCommand \redsequence 	{mm}
	{ \mathrel{\stackunder{\longrightsquigarrow}{\scriptstyle #1,#2}} }

\NewDocumentCommand \redo			{}		
	{ \red^\omega }
\NewDocumentCommand \redog 			{}
	{ \longrightharpoondown^\omega }

\NewDocumentCommand \redord			{o}
	{ \mathrel{
	\stackunder{\Longrightarrow}{\scriptstyle\IfValueT{#1}{#1}}
	}^\infty }

\NewDocumentCommand \muMALL			{}		{ \mu\mathsf{MALL} }
\NewDocumentCommand \muMALLi		{}		{ \muMALL^\infty }
\NewDocumentCommand \muLL			{}		{ \mu\mathsf{LL} }
\NewDocumentCommand \muLLi			{}		{ \muLL^\infty }

\NewDocumentCommand \llneg			{m}		{ {#1}^\bot }
\NewDocumentCommand \normalampersand {}		{{\upshape\mdseries\&}}
\NewDocumentCommand \llpar			{}
	{ \mathop{ \rotatebox[origin=c]{180}{\normalampersand} } }
\NewDocumentCommand \llwith			{}		
	{ \mathop{\text{\normalampersand}} }
\NewDocumentCommand \lltens			{}		{ \mathop{\otimes} }
\NewDocumentCommand \llplus			{}		{ \mathop{\oplus} }

\NewDocumentCommand \cutrel 		{}		{ \sqcup }

\NewDocumentCommand	\indi			{}		{ \mathbf{i} }
\NewDocumentCommand	\indj			{}		{ \mathbf{j} }
\NewDocumentCommand \seqcard		{m}		{ \left|#1\right| }

\NewDocumentCommand \moresequents	{}		{ \vec{\Zeta} }
\NewDocumentCommand \conclusion		{}		{ \Eta }

\NewDocumentCommand \cutelimrule {mm}
	{\proofrule{#1}/\proofrule{#2}}
	
\NewDocumentCommand \cutelimstep {m}
	{\begin{center}%
	\scalebox{.9}{\clap{%
		\setlength{\fboxsep}{10pt}%
		\colorbox{lipicsYellow!50}{\ensuremath{#1}}%
	}}%
	\end{center}}

\NewDocumentCommand \widered {} 
	{ \quad\raisebox{1.5ex}{$\red$}\quad }

\NewDocumentCommand \prempred {m} { \mathrm{pred}(#1) }

\title{Compression for Coinductive Rewriting
and the Cut-Elimination of Non-Wellfounded Proofs}
\titlerunning{Compression for Coinductive Rewriting}

\author{Rémy Cerda}
	{Université Paris Cité, CNRS, IRIF, F-75013, Paris, France
	\and DISI, Università di Bologna, Italy
	\and \url{https://www.i2m.univ-amu.fr/perso/remy.cerda/}}
	{Remy.Cerda@math.cnrs.fr}
	{https://orcid.org/0000-0003-0731-6211}
	{}

\author{Alexis Saurin}
	{Université Paris Cité, CNRS, IRIF, F-75013, Paris, France
	\and INRIA Picube, Paris, France
	\and \url{https://www.irif.fr/users/saurin/index}}
	{Alexis.Saurin@irif.fr}
	{https://orcid.org/0009-0002-1304-5518}
	{}

\authorrunning{R. Cerda and A. Saurin}

\Copyright{Rémy Cerda and Alexis Saurin}



\shortlongversion{%
	\relatedversiondetails[cite=Cerda.Sau.26a]
	{Extended version}
	{https://arxiv.org/abs/2510.08420}%
}{%
	\relatedversiondetails[cite=Cerda.Sau.26]
	{Published version}
	{https://dx.doi.org/10.4230/LIPIcs.MFCS.2026.21}%
}

\funding{The authors were partially funded by the ANR project RECIPROG 
ANR-21-CE48-019.}

\acknowledgements{%
	The authors wish to thank
	Emmanuel Beffara for the \TeX{}nical help
	he kindly provided as they were trying to play
	with his wonderful package
	\href{https://ctan.org/pkg/ebproof}{\tt ebproof},
	Vincent van Oostrom and the participants of HOR~2025
	for constructive discussions about this work,
	and the anonymous reviewers of previous versions of this article
	for their significant contribution to its improvement.
}

\nolinenumbers 

\EventEditors{Michal Kouck\'{y} and Daniela Petri\c{s}an}
\EventNoEds{2}
\EventLongTitle{51st International Symposium on Mathematical Foundations of 
Computer Science (MFCS 2026)}
\EventShortTitle{MFCS 2026}
\EventAcronym{MFCS}
\EventYear{2026}
\EventDate{August 24--28, 2026}
\EventLocation{Paris, France}
\EventLogo{}
\SeriesVolume{386}
\ArticleNo{16}

\begin{document}

\maketitle

%
%

\begin{abstract}
We introduce a generic presentation of
\enquote{syntactic objects built by mixed induction and coinduction} 
encompassing all standard kinds of infinitary terms,
as well as derivation trees in non-wellfounded proof systems.
We then define a coinductive notion of infinitary rewriting of such objects,
which is equivalent to the original presentation of infinitary rewriting 
relying on metric convergence 
and ordinal-indexed sequences of rewriting steps. 
This provides a unified coinductive presentation of 
\eg first-order infinitary rewriting, infinitary λ-calculi, 
and cut-elimination in non-wellfounded proofs.

We then formulate and study the coinductive counterpart of compression, 
\ie the property of an infinitary rewriting system such that
all rewriting sequences of any ordinal length can be \enquote{compressed}
to equivalent sequences of length at most $ω$
(which ensures that they can be finitely approximated).
We characterise compression
in our generic setting for coinductive rewriting, 
\enquote{factorising} the part of the proof that can be performed 
at this level of generality. 
Our proof is fully coinductive, avoiding any detour via rewriting sequences.

Finally we focus on the non-wellfounded proof system $\muMALLi$
for multiplicative-additive linear logic with fixed points,
and we put our results to work in order to prove that 
compression holds for cut-elimination in this setting, 
which is a key lemma of several extensions of cut-elimination
to similar systems.
\end{abstract}


\section{Introduction}
\label{sec:introduction}

Infinite objects and processes pervade the study of computing%
.
It is indeed well-known that
considering non-terminating computations
potentially resulting in undefined results 
is fundamental in the development of any
universal model of computation. 
Non-terminating computations are also crucial in the study
of productive processes, for instance for reactive programs:
in such a setting, programs may run forever but one requires
either a good behaviour with respect to the environment (\emph{fairness})
or that arbitrary approximations of a result are computed
in finite time (\emph{productivity}).

In the presence of non-terminating behaviours,
it is natural to consider ideal objects representing
what is computed at the limit of the process
(typically built \emph{via} a completion 
of the set of finite processes
with respect to some algebraic or topological structure).
Standard examples include infinite streams 
obtained by completing finite words,
or Böhm trees resulting from the ideal completion
of λ-terms extended with an undefined value \cite{Barendregt.77},
that can be seen as a \enquote{syntactic description of the semantics}
of a program.
Once such limit objects enter the picture, 
it makes sense to give them a first-class status
and to allow computation to be performed directly
on infinite objects;
in particular in a functional setting one wants
to ensure compositionality,
\ie that the (possibly infinite) result of a computation should in turn be 
applicable to some arguments, giving rise to a new
(possibly infinite) computation.

This motivated the introduction and the study of \emph{infinitary rewriting}
in the 1990s:
in an infinitary rewriting system
the terms may be infinite  (they result from the metric completion
of usual, finite terms)
and the sequences of rewriting steps may be indexed by arbitrary ordinals.
This idea was successfully applied to
first-order rewriting
\cite{Dershowitz.Kap.Pla.91,Kennaway.Klo.Sle.Vri.95},
to the λ-calculus \cite{Kennaway.Klo.Sle.Vri.97},
to higher-order rewriting \cite{Ketema.Sim.11}, etc.
Therein, the limit objects of interest 
are no longer floating somewhere between syntax and semantics:
they become plain syntactic objects
(typically the normal forms of the infinitary rewriting).

%
%

\myparagraph{Infinitary rewriting for non-wellfounded proof theory}
On the other side of the Curry-Howard correspondence
between programs and proofs,
infinite objects and their dynamics are the subject
of a blooming line of work around non-wellfounded proof systems
(following \cite{Santocanale.02,Sprenger.Dam.03}).
In such systems, proof derivations are potentially infinite trees
subject to a validity criterion ensuring the productivity
of the proof (and hence its correctness).
Such systems are used in particular for modal logics
(\eg an infinite proof can account for the unfolding of a
\enquote{forever} temporal modality)
or for logics extended with fixed point operators
(here an infinite proof typically has a greatest fixed point as a conclusion,
or a least fixed point as a hypothesis).
In practice such systems
(particularly their \emph{circular} or \emph{cyclic} fragments,
whose proof derivations are regular)
are very well-suited to proof search \cite{Tsukada.Unn.22}.

As usual in proof theory, good properties of a non-wellfounded proof system
(like correctness, or more practically the good behaviour
of proof search procedures) are ensured by \emph{cut-elimination},
\ie the fact that the cut rule is admissible
\cite{Afs.Klo.25,Baelde.Dou.Sau.16,Das.Pou.19,Fortier.San.13,Saurin.23}.
A standard way to prove this property is to show that each occurrence
of the cut rule can be \enquote{moved upwards} in the proof derivation:
repeated applications of this fact produce finite approximations
of a limit cut-free derivation.
This can be presented in the form of an infinite sequence
of rewriting steps (acting on the proof derivations);
unfortunately this rewriting of non-wellfounded proofs
does not really fit the usual presentation of infinitary rewriting,
which had to be partially adapted specifically for this purpose 
\cite{Baelde.Dou.Sau.16,Saurin.23}.
This calls for more general definitions encompassing
the infinitary rewriting of both 
infinitary terms and non-wellfounded proofs.

\myparagraph{The compression property}
The \enquote{traditional} line of work on infinitary rewriting 
is based on ordinal-indexed strongly Cauchy convergent rewriting sequences 
(\ie rewriting sequences that do not only converge in the topological sense,
but such that in addition the computation steps occur deeper and deeper
in the rewritten objects). 
A desirable property of such rewriting systems is \emph{compression},
\ie the fact that rewriting sequences of arbitrary ordinal length
can be compressed to sequences of length $\omega$.
For example consider the first-order rewriting rules
$a \red_1 f(g(a))$ and $g(f(x)) \red_2 f(x)$,
then the (strongly converging) rewriting sequence
\begin{equation} \label{eq:intro-example-1}
	a \red_1 f(g(a)) \redi_1 f(g(f(g(\dots))))
   	\red_2 f(f(g(\dots))) \redi_2 f(f(f(\dots)))
\end{equation}
is of ordinal length $\omega \cdot 2$ but can be compressed
by interleaving the rewriting steps:
\begin{equation} \label{eq:intro-example-2}
	a \red_1 {\co f(g(a))} \red_1 f(g(f(g(a))))
   	\red_2 f( {\co f(g(a))} ) \redi_{1,2} f(f(f(\dots))).
\end{equation}
This example illustrates the key benefit of compression:
in a strongly converging rewriting sequence of length~$\omega$,
finite approximations of the limit are computed in finite time,
which is clearly not the case with sequences of bigger ordinal length
(in the example above, it takes $\omega + 1$ steps to produce
the two outermost $f$).
This \enquote{approximation} or \enquote{continuity} property
enjoyed by rewriting sequences of length~$\omega$
is at the heart of most practical motivations for the use
of infinitary rewriting.
For example, in
infinitary λ-calculus,
it is the reason why infinitary rewriting allows to provide
an easy proof of the continuous approximation theorem \cite{Cerda.24},
a result that is of paramount importance for the classical
study of the λ-calculus \cite{Barendregt}.

To the best of our knowledge, compression lemmas have been proved
for three kinds of term rewriting systems:
	left-linear first-order rewriting systems
	\cite{Kennaway.Klo.Sle.Vri.95,Ketema.12,Lombardi.Rio.Vri.14,
	Endrullis.Han.Hen.Pol.Sil.18},
	infinitary λ-calculi
	\cite{Kennaway.Klo.Sle.Vri.97,Bahr.10},
	fully-extended, left-linear higher-order rewriting systems
	\cite{Ketema.Sim.11}.
The compression property also plays a significant role
in non-wellfounded proof theory:
although direct proofs of cut-elimination typically
produce a rewriting sequence of length~$\omega$,
some other cut-elimination results
involve rewriting sequences of arbitrary ordinal length as they 
rely on translations
of formulæ and proof derivations from one logic to another,
where single cut-elimination steps are translated into
sequences of ordinal length~\cite{Saurin.23}.
In these cases the result crucially relies on a compression lemma
ensuring that finite approximations of a cut-free derivation
are produced in finitely many steps.
Again, this variety of compression lemmas
all enjoying more or less the same proof
suggests a generic, uniform treatment.

\myparagraph{The coinductive turn}
As an alternative to the traditional definition of infinite objects
\emph{via} ideal or metric completion,
a more recent and very fruitful line of work
is based on the use of \emph{coinduction}:
the metric completion of any algebraic type (\eg a type of terms)
can indeed be described as
the corresponding coalgebraic type \cite{Barr.93}.
This coinductive reformulation has been in particular 
extensively conducted in the setting of infinitary λ-calculi
\cite{Joachimski04,Endrullis.Pol.13,Czajka.20,Cerda.24}
and has noticeably allowed for a clean presentation of α-equivalence
for infinitary λ-terms \cite{Kurz.Pet.Sev.Vri.13,Cerda.25}.

Although this coinductive redefinition of infinitary terms
is now very standard,
finding a coinductive counterpart to strongly converging
rewriting sequences turns out to be less easy.
A first answer has been proposed for an infinitary λ-calculus
\cite{Endrullis.Pol.13}, 
but in a presentation corresponding solely to sequences of length $\omega$,
and a generic presentation of coinductive infinitary rewriting
was designed only recently by Endrullis, Hansen, Hendriks,
Polonsky and Silva \cite{Endrullis.Han.Hen.Pol.Sil.18}
but is limited to first-order term rewriting.

Once more, this calls for a generic presentation of coinductive rewriting
unifying and completing the aforementioned lines of work.
In particular expressing the infinitary rewriting induced by cut-elimination
on non-wellfounded proofs coinductively seems particularly natural,
and would pave the way for fully coinductive cut-elimination proofs.

\myparagraph{Contributions and organisation of the paper}
In \Cref{sec:infinintary-rewriting}, we propose a framework
for describing arbitrary non-wellfounded objects and rewriting systems
acting on such objects, 
which encompasses all the use cases mentioned so far.
We extend the correspondence between
topology-based and coinduction-based infinitary rewriting
to this generic framework.
In \Cref{sec:compression-lemma}, we formulate a characterisation of 
the compression property \enquote{factorising}
the part of the proof that can be performed at this level of generality,
hence identifying the key features 
a compressible infinitary rewriting system should enjoy;
to our knowledge this proof is new.
We then apply it to first-order coinductive rewriting 
and to the λ-calculus. 
Finally, in \Cref{sec:compressing-mumall} we recall 
the non-wellfounded proof system $\muMALLi$
for multiplicative-additive linear logic with fixed points, 
and we prove that compression also holds for
the rewriting induced by cut-elimination on this proof system.
The choice of $\muMALLi$ is justified by the fact that compression
in this setting is a cornerstone argument of the proof of
cut-elimination for $\muLLi$
(the non-wellfounded proof system for linear logic with fixed points),
in which non-wellfounded proof systems for a whole range of
logics (intuitionistic, classical, temporal, etc.)
can be embedded.
Finally, we recap and suggest directions for future work in 
\Cref{sec:conclusion}.

\section{Infinitary rewriting: From strong convergence to coinduction}
\label{sec:infinintary-rewriting}

In this section, we first recall how infinitary first-order rewriting
based on strong Cauchy convergence \cite{Kennaway.Klo.Sle.Vri.95}
can be reformulated using coinduction,
as exposed in \cite{Endrullis.Han.Hen.Pol.Sil.18},
and we similarly connect
the presentations of infinitary λ-calculi
based on convergence \cite{Kennaway.Klo.Sle.Vri.97}
and on coinduction \cite{Endrullis.Pol.13}.
Finally we provide a generic framework for infinitary rewriting
of non-wellfounded objects,
extending again the correspondence between the two views
on infinitary rewriting.

\subsection{A first example: First-order rewriting}
\label{ssec:fo}

\myparagraph{First-order terms}
Fix a countable set $\Var$ of \defemph{variables}.
A \defemph{first-order signature} is given by a countable set $\Sigma$
of symbols endowed with an arity function $\arity : \Sigma \to \Nat$;
we fix such a signature.
%
	The set $\foterms$ of all \defemph{finite first-order terms}
	on the signature $\Sigma$ is defined by the following
	set of inductive rules:
	\begin{equation} \label{def:foterms}
		\begin{myprooftree}
		\infer0{ x \in \foterms }
		\end{myprooftree}
	\qquad
		\begin{myprooftree}
		\hypo{ s_1 \in \foterms }
		\hypo{ \dots }
		\hypo{ s_{\arity(\focons c)} \in \foterms }
		\infer3
			{ \focons c(s_1, \dots, s_{\arity(\focons c)}) \in \foterms }
		\end{myprooftree}.
	\end{equation}
	(Whenever we give such a set of rules where a rule features
	a variable $x$, a constructor $\focons c$,
	or more generally an element of an external set,
	we assume that there is one rule for each element of this set.)
The \defemph{truncation} of a term $s \in \foterms$ 
at depth $d \in \Nat$ is defined inductively by
	$\trunc[0]{s} \coloneqq \ast$,
	$\trunc[d+1]{x} \coloneqq x$, and
	$\trunc[d+1]{\focons c(s_1,\dots,s_k)} \coloneqq
		\focons c(\trunc{s_1}, \dots, \trunc{s_k})$,
where $\ast$ is a a fresh nullary symbol.
The set $\foterms$ is equipped with a metric $\dist$ defined by
$\dist(s,t) \coloneqq \inf \set{ 2^{-d} }[ \trunc s = \trunc t ]$.
%
	The set $\foiterms$ of all \defemph{(infinitary) first-order terms}
	on the signature $\Sigma$ is the metric completion of $\foterms$
	\wrt $\dist$.
%
As a consequence of \cite[Proposition 3.1]{Barr.93},
$\foiterms$ can be equivalently presented as 
the set defined by treating the rules of \cref{def:foterms} coinductively.
(In short,
$\foterms \coloneqq \mu X.FX$ and $\foiterms \coloneqq \nu X.FX$ where
$FX \coloneqq \Var + \coprod_{\focons c \in \Sigma}
X^{\arity(\focons c)}$ and
the metric completion is carried by the canonical
(co)algebra morphism $\foterms \to \foiterms$.)

\myparagraph{First-order rewriting}
A \defemph{substitution} is any function $\Var \to \foiterms$.
Given a substitution $\sigma$ and a term $t$, we denote by $\sigma \cdot t$
the term  defined coinductively by
$\sigma \cdot x \coloneqq \sigma(x)$ and
$\sigma \cdot \focons c(s_1,\dots,s_k) \coloneqq
\focons c(\sigma \cdot s_1,\dots,\sigma \cdot s_k)$.
%
	A (first-order) \defemph{rewrite rule} is a pair 
	$(l,r) \in \foiterms \times \foiterms$ such that
	\begin{ienumerate}
	\item $l$ is not a variable,
	\item all variables occurring in $r$ also occur in $l$.
	\end{ienumerate}
	An \defemph{infinitary (first-order) term rewriting system},
	\itrs in short, is a set of rewrite rules.
	An \itrs $\mathcal{R}$ defines a rewriting relation $\red$
	on $\foiterms$ by the following set of inductive rules:
	\begin{equation} \label{def:fo-red}
		\begin{myprooftree}
		\hypo{ (l,r) \in \mathcal{R} }
		\hypo{ \sigma : \Var \to \foiterms }
		\infer2{ \sigma \cdot l \redd{0} \sigma \cdot r }
		\end{myprooftree}
	\qquad
		\begin{myprooftree}
		\hypo{ s_i \redd{d} s'_i }
		\hypo{ 1 \leq i \leq \arity(\focons c) }
		\infer2
			{ \focons c(s_1, \dots, s_{\arity(\focons c)}) \redd{d+1}
			\focons c(s_1, \dots, s'_i, \dots, s_{\arity(\focons c)}) }
		\end{myprooftree}
	\end{equation}
	where, as syntactic sugar,
	we sometimes annotate $\red$ with {\color{depth} integers} indicating
	the depth at which the rewriting step occurs.
%
From now on we fix an \itrs $\mathcal{R}$.

\begin{definition} \label{def:sconv-sequences-fo}
	Given an ordinal $\gamma$,
	a \defemph{rewriting sequence} of length $\gamma$ 
	from $s_0$ to $s_\gamma$ is the data of
	terms $(s_\delta)_{\delta \leq \gamma}$ together with rewriting steps 
	$(s_\delta \redd{d_\delta} s_{\delta+1})_{\delta < \gamma}$.
	It is \defemph{strongly converging} if
	for every limit ordinal $\gamma' \leq \gamma$,
	\begin{ienumerate}
	\item $\lim_{\delta \to \gamma'} s_\delta = s_{\gamma'}$ and
	\item $\lim_{\delta \to \gamma'} d_\delta = \infty$.
	\end{ienumerate}
	We write $s \redi t$ whenever there is such a strongly converging
	rewriting sequence (of any ordinal length $\gamma$)
	such that $s_0 = s$ and $s_{\gamma} = t$.
\end{definition}

We now introduce an alternative definition using coinduction.
In general, the sets of rules we introduce to do so 
contain both inductive and coinductive rules;
to distinguish them, we draw double rules for the latter.
In such a system, non-wellfounded derivations are allowed
provided any infinite branch crosses infinitely many
coinductive rules.

\begin{definition} \label{def:redi-fo}
	For all ordinals $\gamma < \omega_1$,
	relations $\redi[\gamma]$ and $\redig[\gamma]$ on $\foiterms$
	are mutually defined by the rules:
	\[
		\begin{myprooftree}
		\hypo{ s \redsequence{\gamma}{m} s' }
		\hypo{ s' \redig[\gamma] t }
		\infer2[split]{ s \redi[\gamma] t }
		\end{myprooftree}
	\qquad
		\begin{myprooftree}
		\infer0[lift_\var{x}]{ \var x \redig[\gamma] \var x }
		\end{myprooftree}
	\qquad
		\begin{myprooftree}
		\hypo{ s_1 \redi[\gamma] s'_1 }
		\hypo{ \dots }
		\hypo{ s_{\arity(\focons c)} \redi[\gamma] s'_{\arity(\focons c)} }
		\infer[double]3[lift_{\focons c}]
			{ \focons c(s_1,\dots,s_{\arity(\focons c)}) 
			\redig[\gamma] \focons c(s'_1,\dots,s'_{\arity(\focons c)}) }
		\end{myprooftree}
	\]
	where the notation $s \redsequence{\gamma}{m} s'$
	stands for any sequence
	$	s \reds s'_1 \redig[\delta_1] t_1 \reds s'_2 \redig[\delta_2] \dots
		\redig[\delta_m] t_m \reds s' $
	such that $\forall 1 \leq i \leq m,\ \delta_i < \gamma$.
\end{definition}

Observe that $\redig[\gamma]$ is just $\redi[\gamma]$
under a constructor, which enforces strong convergence.
This Definition and the following Theorem are instantiated
on the example rewriting sequence from \cref{eq:intro-example-1}
in \cref{app:example}, \cref{app:example:eq:1}.

\begin{theorem}[{\cite[Theorem 5.2]{Endrullis.Han.Hen.Pol.Sil.18}}]
\label{the:M-equiv-C-fo}
	$\redi$ is the union of all relations $\redi[\gamma]$
	(over $\gamma < \omega_1$).
\end{theorem}

Let us mention that we marginally depart from the original definition 
\cite[Definition~4.2]{Endrullis.Han.Hen.Pol.Sil.18}
by annotating the relations with ordinals and 
ensuring that these annotating ordinals decrease
along certain branches of the derivations,
whereas the original authors explicitly add 
the (clearly equivalent) constraint that 
no branch should cross infinitely many $\longrightsquigarrow$.
Relaxing this constraint would give rise to \enquote{bi-infinite} rewriting,
corresponding to rewriting sequences that may also be transfinite
\enquote{to the left}, \ie towards their source.
In particular our ordinal annotations \emph{do not correspond}
to the ordinal length of a corresponding rewriting sequence!
(We can only say that any rewriting sequence of length $\gamma$
can be turned into $\redi[\gamma]$, see \cref{the:M-equiv-C}.)
We also replace their rule \proofrule{id},
which adds $s \redig[\gamma] s$ as an axiom,
by its version \proofrule{lift_{\var x}} restricted to variables.
This is enough for the defined relations to be reflexive
(see \cref{lem:redig-reflexive}).

\subsection{A second example: Infinitary λ-calculi}
\label{ssec:lambda}

\myparagraph{λ-terms}
The construction of infinitary λ-terms follows the same path,
let us summarise it.
The set $\lterms$ of \defemph{finite λ-terms} is defined 
inductively by:
$s,t,\dots\ \coloneqq\ x \ |\ λx.s \ |\ st$ (where $x\in\Var$).
Fix $a,b,c \in \Bool^3$. The $abc$-truncation of $s \in \lterms$
at depth $d \in \Nat$ is defined inductively by
$\trunc[0]{s}[abc] \coloneqq \ast$, and by
$\trunc[d+1]{x}[abc] \coloneqq x$,
$\trunc[d+1]{λx.s}[abc] \coloneqq λx.\trunc[d+1-a]{s}[abc]$ and
$\trunc[d+1]{st}[abc] \coloneqq 
	\trunc[d+1-b]{s}[abc] \trunc[d+1-c]{t}[abc]$.
For instance, taking $abc = 001$ means that we consider that
one goes \enquote{deeper} in a λ-term only when entering
the argument of an application.
We equip $\lterms$ with a metric $\dist[abc]$ 
defined by $\dist[abc](s,t) \coloneqq \inf 
\set{ 2^{-d} }[ \smash{\trunc s [abc] = \trunc t [abc]} ]$.
The set $\liiiterms abc$ of \defemph{$abc$-infinitary λ-terms}
is defined as the metric completion of $\lterms$
\wrt $\dist[abc]$.

Equivalently, the following coinductive definition of $\liiiterms abc$
has been proposed \cite{Cerda.24}:
\begin{equation} \label{eq:liiiterms}
	\begin{myprooftree}[center]
	\hypo{\strut}
	\infer1{ x \in \liiiterms abc }
	\end{myprooftree}
\qquad
	\begin{myprooftree}[center]
	\hypo{ \later_a s \in \liiiterms abc }
	\infer1{ λx.s \in \liiiterms abc }
	\end{myprooftree}
\qquad
	\begin{myprooftree}[center]
	\hypo{ \later_b s \in \liiiterms abc }
	\hypo{ \later_c t \in \liiiterms abc }
	\infer2{ st \in \liiiterms abc}
	\end{myprooftree}
\qquad
	\begin{myprooftree}[center]
	\hypo{ S }
	\infer1[\later_0]{ \later_0 S}
	\end{myprooftree}
\qquad
	\begin{myprooftree}[center]
	\hypo{ S }
	\infer[double]1[\later_1]{ \later_1 S}
	\end{myprooftree}
\end{equation}
where $S$ stands for any statement
(in this case, of the form \enquote{$s \in \liiiterms abc$}).
While $\later_0$ is an \enquote{inductive guard}
(hence could as well be omitted),
$\later_1$ introduces a coinductive guard
(it is the \enquote{later} modality from
\cite{Nakano.00,Appel.Mel.Ric.Vou.07}):
$a = 1$ means that λ-abstraction is coinductive,
while $a = 0$ means that it is inductive
--- and similarly for $b$ and $c$
for the left and right sides of application.
(In short\footnotemark, using fixed points again,
$\lterms \coloneqq \mu X.F^{abc}(X,X)$ and
$\liiiterms abc \coloneqq \nu X_1.\mu X_0.F^{abc}(X_0,X_1)$
where $F^{abc}(X_0,X_1) \coloneqq 
\Var + \Var \times X_a + X_b \times X_c$.)
%
\footnotetext{%
\label{rem:alpha}%
In this summary we carefully omitted to mention
any quotient by α-equivalence (\ie by renaming of bound variables).
A rigorous construction can be found in \cite{Kurz.Pet.Sev.Vri.13,Cerda.25};
it consists in working in the category of nominal sets and defining
$F^{abc}(X_0,X_1) \coloneqq \Var + [\Var]X_a + X_b \times X_c$,
where $[\Var]$ is a functor encoding \enquote{nameless} abstraction.
We do not detail these technicalities,
as all the following work could be straightforwardly transported
from the naive presentation above to the rigorous one.
}

\myparagraph{β-reduction}
We denote by $s \subst t$ the λ-term obtained by substituting all free
occurrences of $x$ with $t$ in $s$, in a capture-avoiding manner%
.
%
	\defemph{β-reduction} is the relation on $\liiiterms abc$
	defined inductively by:
	\[
		\begin{myprooftree}
		\infer0{ (λx.s)t \redd{0} s \subst t }
		\end{myprooftree}
	\qquad
		\begin{myprooftree}
		\hypo{ u \redd{d} u' }
		\infer1{ λx.u \redd{d+a} λx.u' }
		\end{myprooftree}
	\qquad
		\begin{myprooftree}
		\hypo{ u \redd{d} u' }
		\infer1{ uv \redd{d+b} u'v }
		\end{myprooftree}
	\qquad
		\begin{myprooftree}
		\hypo{ v \redd{d} v' }
		\infer1{ uv \redd{d+c} uv' }
		\end{myprooftree}
	\]
	where again we sometimes annotate $\red$ with
	the depth at which the rewriting step occurs.
	As in \cref{def:sconv-sequences-fo}, we write $s \redi t$
	whenever there is a strongly converging β-reduction sequence
	from $s$ to $t$.
%
The following counterpart of \cref{the:M-equiv-C-fo}
states that this topological presentation
can be turned into a coinductive one.
This is another particular case of \cref{the:M-equiv-C},
to be stated in the next section.

\begin{theorem} \label{the:M-equiv-C-lambda}
	$\redi$ is the union of all relations $\redi[\gamma]$
	(over ordinals $\gamma < \omega_1$) defined by 
	the rules \proofrule{split} and \proofrule{lift_\var{x}}
	from \cref{def:redi-fo},
	the rules \proofrule{\later_0} and \proofrule{\later_1}
	from \cref{eq:liiiterms}, as well as the rules
	$\ 
		\begin{myprooftree}[center]
		\hypo{ \later_a u \redi[\gamma] u' }
		\infer1[lift_λ]{ λx.u \redig[\gamma] λx.u' }
		\end{myprooftree}
	\ $ and $\ 
		\begin{myprooftree}[center]
		\hypo{ \later_b u \redi[\gamma] u' }
		\hypo{ \later_c v \redi[\gamma] v' }
		\infer2[lift_@]{ uv \redig[\gamma] u'v'. }
		\end{myprooftree}
	\ $.
\end{theorem}

\subsection{Infinitary rewriting of many-sorted infinitary terms}
\label{ssec:arbitrary-derivations}

We will now introduce a description of rewriting for arbitrary
non-wellfounded objects, in the form of derivation trees
that we call \enquote{many-sorted terms}.
We want in particular to mix not only inductive constructors
and coinductive constructors,
but to feature constructors with mixed inductive and coinductive inputs,
as displayed by the example of $abc$-infinitary λ-calculi;
and instead of manipulating the rather heavy
$\later_0$ and $\later_1$ modalities and the corresponding rules,
we introduce the following kind of mixed rules:
\begin{equation*}
	\begin{myprooftree}
	\infer0{ x \in \liiiterms 001 }
	\end{myprooftree}
\qquad
	\begin{myprooftree}
	\hypo{ s \in \liiiterms 001 }
	\infer1{ λx.s \in \liiiterms 001 }
	\end{myprooftree}
\qquad
	\begin{myprooftree}
	\hypo{ s \in \liiiterms 001 }
	\hypo[doubleprem]{ t \in \liiiterms 001 }
	\infer2{ st \in \liiiterms 001}
	\end{myprooftree}
\end{equation*}
that are equivalent (here for defining $\liiiterms 001$)
to the system from \cref{eq:liiiterms}.

In general,
we first fix a set $\mathcal S$ of \defemph{sorts}.
It is typically an inductively defined set using
the singleton type, one or more alphabets,
as well as tuples, lists, multisets, etc.
For instance one could consider 
the set of two-sided sequents in a given logic,
encoded as the set of pairs of lists of formulæ.

\begin{definition} \label{def:dtree}
	A \defemph{construction rule} \proofrule{r} is given by
	\begin{ienumerate}
	\item its arity $\arity(\drule r) \in \Nat$,
	\item a partial function $\drule r : \mathcal S^{\arity(\drule r)}
		\partialto \mathcal S$ mapping its \defemph{premises}
		to its \defemph{conclusion},
	\item a map $\coind : [1,\arity(\drule r)] \to \Bool$
		indicating the (co)inductive
		character of each premise.
	\end{ienumerate}
	It is represented by:
	$\ 
		\begin{myprooftree}[center]
		\hypo[ordoubleprem]{ S_1 }
		\hypo{ \dots }
		\hypo[ordoubleprem]{ S_{\arity(\drule r)} }
		\infer3[r]{ \drule r(S_1,\dots,S_{\arity(\drule r)}) }
		\end{myprooftree}
	\ $
	where the dashed line under $S_i$ means a full line
	whenever $\coind(i) = 1$,
	and an absence of a line otherwise.
	(When we apply such a rule $\drule r$ to some arguments
	we implicitly suppose that they belong
	to its domain of definition.)
	We denote by $\dtrees$ the set of all (non-wellfounded) 
	derivation trees generated by 
	a family $\mathcal D$ of construction rules
	such that all infinite branches cross infinitely many
	double lines (\ie coinductive premises),
	that we call \defemph{many-sorted terms}.
\end{definition}

\begin{examples} \label{exa:dtree}
\begin{enumerate}
\item Pre-proofs in your favourite non-wellfounded proof system, \eg 
	for some logic with fixed points,
	can be presented as an instance of this system
	(see \cref{ssec:muMALLi} for a detailed presentation of 
	the non-wellfounded system $\muMALLi$ for multiplicative-additive
	linear logics with fixed points).
\item The set $\foiterms$ of first-order terms on the signature $\Sigma$
	can be seen as the many-sorted terms
	generated by $\mathcal S \coloneqq \set{\bullet}$ and
	$\mathcal D_\Sigma \coloneqq 
	\set{ \mathrm{Var}_x : \mathcal S^0 \to \mathcal S }[ x \in \Var ] \cup 
	\set{ \mathrm{Cons}_{\focons c} : \mathcal S^{\arity(\focons c)}
	\to \mathcal S }[ \focons c \in \Sigma ]$, together with
	$\coind[\mathrm{Cons}_{\focons c}](i) \coloneqq 1$
	for all $\focons c$ and $i$.
\item Similarly, $abc$-infinitary λ-terms can be seen as
	the many-sorted terms generated by 
	$\mathcal S \coloneqq \set{\bullet}$ and
	$\mathcal D_{\Lambda}^{abc} \coloneqq 
	\set{ \mathrm{Var}_x : \mathcal S^0 \to \mathcal S }[ x \in \Var ] \cup 
	\set{ \mathrm{Abs}_x : \mathcal S \to \mathcal S }[ x \in \Var ] \cup 
	\set{ \mathrm{App} : \mathcal S^2 \to \mathcal S }$,
	together with $\coind[\mathrm{Abs}_x](1) \coloneqq a$,
	$\coind[\mathrm{App}](1) \coloneqq b$ and
	$\coind[\mathrm{App}](2) \coloneqq c$.
\end{enumerate}
\end{examples}


	By abuse of notation, we write $s = \drule r(s_1,\dots,s_k)$
	to express that 
	\begin{ienumerate}
	\item the last construction rule of the many-sorted term $s$ is 
		\proofrule{r},
		so that there are sorts $S_1,\dots,S_k \in \mathcal S$
		such that $s$ has conclusion $\drule r(S_1,\dots,S_k)$,
	\item each derivation $s_i$ has conclusion $S_i$ and
		is the subtree rooted at the $i$th
		premise of the concluding \proofrule{r}.
	\end{ienumerate}
%
	Using this notation,
	a definition of $\dtrees$ using fixed points is given by
	$	\dtrees \coloneqq \nu X_1.\mu X_0.\coprod_{\drule r \in \mathcal D}
		\drule r\left( X_{\coind(1)}, \dots, X_{\coind(\arity(\drule r))} 
		\right). $
	This construction can also be performed in richer categories
	than the category of sets, \eg the category of nominal sets,
	as already suggested in \cref{rem:alpha}.
	(Concentrating on such a \enquote{nu-mu} data type
	is what allows to express all the notions of rewriting we have in mind;
	one could however extend our work to any alternation of fixed points.)
%
In the following, we fix a set $\mathcal D$
as in \cref{def:dtree}.

\begin{definition}
	For each $S \in \mathcal S$, we define a nullary construction rule
	$\drule{trunc_{\mathnormal S}} : \mathcal S^0 \to \mathcal S$,
	$() \mapsto S$, \ie a rule adding $S$ as an axiom.
	The \defemph{truncation} at depth $d \in \Nat$ 
	of a many-sorted term $s \in \dtrees$ with conclusion $S \in \mathcal S$
	is defined inductively by
	$	\trunc[0]{s} \coloneqq \drule{trunc_{\mathnormal S}}()$
	and $\trunc[d+1]{\drule r(s_1,\dots,s_k)} \coloneqq \drule r(
			\trunc[d+1-\coind(1)]{s_1}, \dots, \trunc[d+1-\coind(k)]{s_k}).$
	The set $\dtrees$ is equipped with a metric $\dist$ defined by
	$\dist(s,t) \coloneqq \inf \set{ 2^{-d} }[ \trunc s \simeq \trunc t ]$,
	where $\simeq$ is the equivalence relation inductively generated
	on truncated many-sorted terms
	by equating all rules $\drule{trunc_{\mathnormal S}}$.
\end{definition}

\begin{definition} \label{def:red}
	A set $\mathord{\red} \subseteq \dtrees \times \dtrees$
	of \defemph{root steps} generates a relation $\red$ by
	the following set of inductive rules:
	\[	\begin{myprooftree}
		\hypo{ s \red t }
		\infer1{ s \redd 0 t }
		\end{myprooftree}
	\qquad
		\begin{myprooftree}
		\hypo{ s_i \redd d s'_i }
		\hypo{ 1 \leq i \leq \arity(\drule r) }
		\infer2{ \drule r(s_1,\dots,s_i,\dots,s_{\arity(\drule r)})
			\redd{d + \coind(i)}
			\drule r(s_1,\dots,s'_i,\dots,s_{\arity(\drule r)}) }
		\end{myprooftree} \]
	where, as in \cref{def:fo-red}, we sometimes annotate $\red$ with
	the depth at which the rewriting step occurs.
	As in \cref{def:sconv-sequences-fo}, we write $s \redi t$
	whenever there is a strongly converging reduction sequence
	from $s$ to $t$.
\end{definition}


The constructions of first-order rewriting in a given \itrs
(\cref{ssec:fo}) and of β-reduction (\cref{ssec:lambda})
are particular cases of this definition applied to
the many-sorted terms generated by the sets of construction rules 
$\mathcal D_\Sigma$
and $\mathcal D_{\Lambda}^{abc}$ from \cref{exa:dtree}, respectively.
(In particular the root steps for a given \itrs are
given by the first rule of \cref{def:fo-red},
namely by the closure of all rewrite rules under substitutions.)
\Cref{the:M-equiv-C-fo,the:M-equiv-C-lambda}, which provided
a coinductive presentation of strongly converging rewriting sequences
in these particular settings, can also be extended
to the general framework we introduced.

\begin{definition} \label{def:redi}
	For all ordinals $\gamma < \omega_1$,
	a relation $\redi[\gamma]$ on $\dtrees$
	is defined by the following rules:
	\[	\begin{myprooftree}
		\hypo{ s \redsequence{\gamma}{m} s' }
		\hypo{ s' \redig[\gamma] t }
		\infer2[split]{ s \redi[\gamma] t }
		\end{myprooftree}
	\qquad
		\begin{myprooftree}
		\hypo[ordoubleprem]{ s_1 \redi[\gamma] s'_1 }
		\hypo{ \dots }
		\hypo[ordoubleprem]
			{ s_{\arity(\drule r)} \redi[\gamma] s'_{\arity(\drule r)} }
		\infer3[lift_r]{ \drule r(s_1,\dots,s_{\arity(\drule r)}) 
			\redig[\gamma] \drule r(s'_1,\dots,s'_{\arity(\drule r)}) }
		\end{myprooftree} \]
	where a rule \proofrule{lift_r} is defined 
	for each $\drule r \in \mathcal D$, 
	and its $i$th premise is coinductive \ifff $\coind(i) = 1$.
\end{definition}

\begin{theorem} \label{the:M-equiv-C}
	$\redi$ is the union of all relations $\redi[\gamma]$
	(over ordinals $\gamma < \omega_1$).
\end{theorem}
	
	\begin{proof}[Proof sketch]
	The proof is as in \cite[Thm.~5.2]{Endrullis.Han.Hen.Pol.Sil.18}.
	If there is a strongly convergent sequence of length $\gamma$
	from $s$ to $t$ (denote it by $s \redord[\gamma] t$),
	we build a derivation $s \redi[\gamma] t$
	by (co)induction, using the fact that only finitely many rewriting steps
	in the sequence are not performed
	above a coinductive premise of a construction rule,
	so that the sequence can be written:
	$	s \reds s'_1 \redord[\delta_1] t_1 \reds s'_2 \redord[\delta_2] \dots
		\redord[\delta_{m+1}] t $
	with $\forall 1 \leq i \leq m,\ \delta_i < \gamma$
	and $\delta_{m+1} \leq \gamma$,
	which exactly gives the premises of the rule \proofrule{split}.
	
	Conversely, the proof proceeds by well-founded induction on $\gamma$
	such that $s \redi[\gamma] t$. The induction hypothesis implies that
	$\mathord{\redsequence{\gamma}{m}} \subseteq \mathord{\redord}$,
	for all $m$, hence we obtain $s \redord s_1 \redig[\gamma] t$.
	By iterating in the inductive premises of $s_1 \redig[\gamma] t$,
	we obtain $s_1 \redord s'_1 \redig[\gamma] t$.
	Then we iterate in the coinductive premises of $s'_1 \redig[\gamma] t$,
	building $s'_1 \redord s'_2 \redord \dots$.
	Since this sequence does not contain any $\redd 0$ any more,
	we can prove strong convergence.
	\end{proof}

\section{A generic compression lemma}
\label{sec:compression-lemma}

As recalled in the introduction, the compression lemma is a result about
several kinds of infinitary rewriting systems stating
that strongly convergent rewriting sequences
of arbitrary ordinal length can be \enquote{compressed}
into sequences of length $\omega$ with the same source and target.
The goal of this section is to characterise this compression property
in the general coinductive framework we just introduced.
To do so, we first provide a coinductive counterpart to
strongly convergent sequences of length $\omega$,
then we prove the main result of the paper (\cref{the:compression}),
a characterisation identifying
\enquote{just what needs to be adapted to each system}
for compression to hold.
We conclude by applying the theorem to
our running examples, first-order rewriting and λ-calculus.
(Let us also already mention that \cref{app:example}
details the proof for the concrete example we gave in 
\cref{eq:intro-example-1,eq:intro-example-2}
in the introduction.)

Whereas \cref{the:M-equiv-C} above was an extension of 
\cite{Endrullis.Han.Hen.Pol.Sil.18}
to our newly introduced \enquote{many-sorted terms},
our characterisation of the compression property is
completely new to our knowledge.
In the unique coinductive proof of compression in the literature,
namely the Coq/Rocq proof of \cite{Endrullis.Han.Hen.Pol.Sil.18}
(which is limited to first-order rewriting,
hence only corresponds to our \cref{the:fo-Q})
the slightly different definition of $\redi$ 
featuring least and greatest fixed points to constrain
the use of coinduction (where we preferred ordinal annotations)
results in a completely different treatment of 
the inductive part of the proof.
All other existing proofs of compression properties
are formulated in a topology-based setting
(usually using the same transfinite induction
as in \cite{Kennaway.Klo.Sle.Vri.95},
which cannot be straightforwardly translated coinductively
since \enquote{the ordinal length of a rewriting sequence}
has no clear coinductive counterpart);
they do of course yield an indirect proof in the coinductive setting,
thanks to our \cref{the:M-equiv-C}.
Thus, besides its generality,
another advantage of our approach with respect to all those is
to provide a direct, explicit coinductive procedure for
compressing derivations of infinitary rewritings.

In the following
we again fix sets $\mathcal S$ and $\mathcal D$ as in \cref{def:dtree},
and a rewriting relation $\red$ on $\dtrees$ as in \cref{def:red}.

\subsection{Compressed rewriting sequences coinductively}

We first
need to give a coinductive presentation
of strongly converging rewriting sequences of length at most $\omega$.
This is quite straightforward,
and is an adaption of \cite[Equation~9.3]{Endrullis.Han.Hen.Pol.Sil.18}.

\begin{definition} \label{def:redo}
	The relation $\redo$ of \defemph{compressed infinitary rewriting}
	is defined on $\dtrees$ by 
	$\mathord{\redo} \coloneqq \mathord{\redi[0]}$,
	\ie it is defined by the following rules \proofrule{split^\omega}
	and, for each $\drule r \in \mathcal D$, \proofrule{lift_r^\omega}:
	\[
		\begin{myprooftree}
		\hypo{ s \reds s' }
		\hypo{ s' \redog t }
		\infer2[split^\omega]{ s \redo t }
		\end{myprooftree}
	\qquad
		\begin{myprooftree}
		\hypo[ordoubleprem]{ s_1 \redo s'_1 \vphantom{s'_(} }
		\hypo{ \dots }
		\hypo[ordoubleprem]
			{ s_{\arity(\drule r)} \redo s'_{\arity(\drule r)} }
		\infer3[lift_r^\omega]{ \drule r(s_1,\dots,s_{\arity(\drule r)}) 
			\redog \drule r(s'_1,\dots,s'_{\arity(\drule r)}) }
		\end{myprooftree}
	\]
	where, as in \cref{def:dtree}, the  $i$th premise of 
	\proofrule{lift_r^\omega} is coinductive \ifff $\coind(i) = 1$.
\end{definition}

\begin{lemma}
	For $s,t \in \dtrees$, $s \redo t$ \ifff 
	there is a strongly converging sequence of length at most $\omega$
	from $s$ to $t$.
\end{lemma}
%
	We say that $\red$ has the \defemph{compression property}
	if $\mathord{\redi} = \mathord{\redo}$.

\subsection{A characterisation of the compression property}

We reach the core technical content of this paper,
where we provide a characterisation of compression
for any rewriting system presented in the generic framework
introduced in the previous section.
We state this result just after the following useful observations
(which are all straightforward consequences of \cref{def:redi}).

\begin{lemma}
\begin{thmenumerate}*{lemma}
\item If $s \redig[\gamma] t$ and $\delta \geq \gamma$,
	then $s \redig[\delta] t$.
	\label[lemma]{lem:redig-higher-ordinals}
\item The relation $\mathord{\redig[\gamma]}$ is reflexive.
	\label[lemma]{lem:redig-reflexive}
\item If $s \redig[\gamma] t \redig[\delta] u$,
	then $s \redig[\epsilon] u$ 
	for $\epsilon \coloneqq \max(\gamma+1, \delta)$.
	\label[lemma]{lem:redig-max-transitivity}
\item Items (1)--(3)
	also hold with $\redi[\gamma]$, $\redi[\delta]$
	instead of $\redig[\gamma]$, $\redi[\delta]$.
	\label[lemma]{lem:redi-technical}
\item As a consequence, the relation $\mathord{\redi}$
	is reflexive and  transitive.
	\label[lemma]{lem:redi-transitivity}
\item The inclusions
	$\mathord{\redig[\delta]} \subset \mathord{\redi[\delta]}$ hold
	for every ordinal $\delta$,
	hence also $\mathord{\redig} \subset \mathord{\redi}$.
	\label[lemma]{lem:redig-included-in-redi}
\end{thmenumerate}
\end{lemma}


\begin{theorem} \label{the:compression}
	We define the following properties of the rewriting relation $\red$:
	\begin{align*}
		\property{P}_\delta: & \quad \forall n \in \Nat,\ 
			\forall s,s' \in \dtrees,\wide
			s \redsequence{\delta}{n} s' \wide[\Rightarrow]
			\exists s'' \in \dtrees,\ \exists \epsilon < \delta,\wide 
			s \reds s'' \redigs[\epsilon] s', \\
		\property{Q}: & \quad \forall \delta,\ 
			\forall s,t,t' \in \dtrees, \wide
			\property P_\delta
			\wide[\wedge] s \redi[\delta] t \red t'
			\wide[\Rightarrow] \exists s' \in \dtrees,\wide s \reds s' 
				\redi[\delta] t'.
	\end{align*}
	Then $\red$ has the compression property \ifff
	the property $\property{Q}$ holds.
\end{theorem}

The property $\property Q$ seems quite convoluted at first sight,
but is in fact not very surprising:
it essentially means that
given a certain induction hypothesis
(namely $\property{P}_\delta$),
a rewriting sequence of ordinal length $\delta + 1$
(for an arbitrary infinite ordinal $\delta$)
can be turned into an equivalent sequence of length $p + \delta = \delta$
for some $p \in \Nat$
(recall that \cref{the:M-equiv-C} translates rewriting sequences
of length $\delta$ into $\redi[\delta]$).
When one tries to prove a compression lemma in a traditional way,
using a transfinite induction and topological arguments,
the key case of the proof is exactly the same,
namely the case of a successor ordinal \cite[§~12.7]{Terese}.

We now give the proof of \cref{the:compression},
which starts with the following key lemmas.

\begin{lemma} \label{lem:P-implies-preponement-root-steps}
	For all ordinal $\delta$
	such that $\property{P}_\delta$ holds
	and for all many-sorted terms $s,t \in \dtrees$,
	if $s \redi[\delta] t$
	then there exists $s' \in \dtrees$
	such that $s \reds s' \redig[\delta] t$.
\end{lemma}
	
	\begin{proof}
	The derivation of $s \redi[\delta] t$
	ends with the rule \proofrule{split},
	with premises $s \redsequence \delta n t'$ and $t' \redig[\delta] t$
	(for some $n \in \Nat$ and $t' \in \dtrees$).
	By $\property{P}_\delta$ there are $s' \in \dtrees$
	and $\epsilon < \delta$ such that $s \reds s' \redigs[\epsilon] t'$,
	and we can conclude by
	\cref{lem:redig-max-transitivity,lem:redig-higher-ordinals}.
	\end{proof}

\begin{lemma} \label{lem:Q-implies-P}
	If $\property Q$ holds then
	$\property{P}_\gamma$ holds for any ordinal $\gamma$.
\end{lemma}
	
	\begin{proof}
	Instead of $\property Q$, 
	we suppose the following property $\property Q'$:
	\[	\forall \delta,\ 
		\forall s,t,t' \in \dtrees, \wide
		\property P_\delta
		\wide[\wedge] s \redig[\delta] t \red t'
		\wide[\Rightarrow] \exists s' \in \dtrees,\wide s \reds s' 
			\redig[\delta] t'. \]
	The only difference with the definition of $\property Q$
	in \cref{the:compression} is that the two occurrences of $\redi[\delta]$
	have been replaced with $\redig[\delta]$.
	This is straightforwardly weaker by
	\cref{lem:redig-included-in-redi} (for the first occurrence)
	and \cref{lem:P-implies-preponement-root-steps} (for the second one).
	
	Then we proceed by well-founded induction on $\gamma$, 
	\ie we suppose that 
	$\property{P}_\delta$ holds for any ordinal $\delta < \gamma$
	and  we	prove that $\property P_{\gamma}$ holds, that is to say:
	\[	\forall m \in \Nat,\ 
		\forall s,s' \in \dtrees,\wide
		s \redsequence{\gamma}{m} s' \wide[\Rightarrow]
		\exists s'' \in \dtrees,\ \exists \delta < \gamma,\wide 
		s \reds s'' \redigs[\delta] s'. \]
	We do this by induction on $m \in \Nat$.
	Take $s,s' \in \dtrees$ such that $s \redsequence{\gamma}{m} s'$.
	We want to build $s'' \in \dtrees$ such that
	$\exists \delta < \gamma,\ s \reds s'' \redigs[\delta] s'$.
	If $m = 0$ the result is immediate: observe that
		$s \redsequence{\gamma}{0} s'$ means that $s \reds s'$
		and use \cref{lem:redig-reflexive}.
	Otherwise, $s \redsequence{\gamma}{m} s'$
		can be decomposed as follows:
		$ 	s \redsequence{\gamma}{m-1} s'_m 
			\redig[\delta_m] t'_m \reds s', $
		with $\delta_m < \gamma$.
		By iterated applications of $\property Q'$, using 
		the induction hypothesis on $\delta_m$
		(\ie the fact that $\property P_{\delta_m}$ holds),
		there is an $s''_m \in \dtrees$ such that:
		$ 	s \redsequence{\gamma}{m-1} s'_m 
			\reds s''_m \redigs[\delta_m] s'$. 
		Notice that $s \redsequence{\gamma}{m-1} s'_m \reds s''_m$
		can be reformulated as 
		$s \redsequence{\gamma}{m-1} s''_m$,
		whence we can apply the induction hypothesis on $m-1$ and obtain
		a term $s''$ and an ordinal $\epsilon < \gamma$ such that:
		$	s \reds s'' \redigs[\epsilon] s''_m \redigs[\delta_m] s'$, 
		which can be simplified as
		$	s \reds s'' \redigs[\delta] s' $
		with $\delta \coloneqq \max(\epsilon,\delta_m) < \gamma$,
		thanks to \cref{lem:redig-higher-ordinals}.
	\end{proof}
	
	\begin{proof}[Proof of \cref{the:compression}]
	Assume $\property{Q}$ holds.
	We want to design a procedure using a derivation of
	$s \redi t$ to coinductively produce a derivation of $s \redo t$.
	To formalise this, let us add the following rule
	to \cref{def:redo}:
	$\ 	\begin{myprooftree}[center]
		\hypo{ s \redi[\gamma] t }
		\infer[snake]1[comp]{ s \redo t }
		\end{myprooftree} \ $,
	and show that any derivation of $s \redo t$
	can be coinductively turned into a derivation of $s \redo t$
	that does not use this rule \proofrule{comp},
	\ie that the rule is admissible.
	Observe that \proofrule{comp} can occur at most once in
	any branch of a derivation, so that
	its occurrences delimit a compressed prefix
	and subtrees remaining to be compressed;
	we have to prove
	\begin{ienumerate}
	\item that in this situation the compressed prefix can be extended
		wherever the rule \proofrule{comp} occurs
		(which we do by analysing the remaining subtrees),
		so that iterating the process
		results in an infinitary compressed derivation
		(as \proofrule{comp} will be pushed to infinity,
		see \cref{app:example} for an illustration), and
	\item that this infinitary derivation is valid,
		\ie each of its infinite branches crosses infinitely often
		a coinductive premise of a rule \proofrule{lift_r^\omega}.
	\end{ienumerate}
	
	We start from a derivation
	whose conclusion $s \redo t$ is obtained by rule \proofrule{comp},
	with hypothesis $s \redi[\gamma] t$.
	The latter can only be obtained
	through the rule \proofrule{split},
	hence there are $s' \in \dtrees$ and $m \in \Nat$ such that
	$s \redsequence{\gamma}{m} s' \redig[\gamma] t$.
	Then by \cref{lem:Q-implies-P}, $\property P_\gamma$ holds,
	hence $s \reds s'' \redigs[\delta] s' \redig[\gamma] t$ 
	for some $s'' \in \dtrees$ and some ordinal $\delta < \gamma$.
	We apply \cref{lem:redig-max-transitivity}
	and obtain $s \reds s'' \redig[\gamma] t$.
	In particular we have the derivation below left
	for some $\drule{r} \in \mathcal{D}$,
	hence we can form the derivation below right:
	\[\begin{array}{l|l}
		\begin{myprooftree}
		\hypo{\vdots}
		\infer[no rule,ordoubleprem]1{ s''_1 \redi[\gamma] t_1 }
		\hypo{\dots}
		\hypo{\vdots}
		\infer[no rule,ordoubleprem]1{ s''_k \redi[\gamma] t_k }
		\infer3[lift_r]{ s'' = \drule r(s_1,\dots,s_{\arity(\drule r)}) 
			\redig[\gamma] \drule r(s'_1,\dots,s'_{\arity(\drule r)}) = t }
		\end{myprooftree}
	&
		\begin{myprooftree}
		\hypo{ s \reds s'' }
		\hypo{\vdots}
		\infer[no rule]1{ s''_1 \redi[\gamma] t_1 }
		\infer[snake,ordoubleprem]1[comp]{ s''_1 \redo t_1 }
		\hypo{\dots}
		\hypo{\vdots}
		\infer[no rule]1{ s''_k \redi[\gamma] t_k }
		\infer[snake,ordoubleprem]1[comp]{ s''_k \redo t_k }
		\infer3[lift_r^\omega]{ s'' \redog t }
		\infer2[split^\omega]{ s \redo t }
		\end{myprooftree}
	\end{array}\]
	and we iterate the process coinductively in each subtree
	rooted at $s''_i \redo t_i$.
	In addition the derivation we produce is correct
	(\ie every of its infinite branches 
	crosses infinitely many coinductive premises)
	because its structure exactly follows the one of the target $t$:
	each \proofrule{split^\omega} followed by \proofrule{lift_r^\omega}
	in the obtained derivation
	corresponds to an \proofrule{r} in the derivation $t$,
	hence the validity of the former is ensured
	by the validity of the latter.
	
	Conversely, if the compression property holds then
	any reduction $s \redi[\delta] t \red t'$
	can be compressed to $s \redo t'$.
	This is equivalent to $s \redi[0] t'$,
	hence $s \redi[\delta] t'$ by \cref{lem:redig-higher-ordinals}.
	Finally, $s \reds s \redi[\delta] t'$.
	\end{proof}

\subsection{Compression for left-linear (coinductive) first-order rewriting}

In this subsection, we fix an \itrs $\mathcal R$ as in \cref{ssec:fo},
and we prove the compression lemma for first-order rewriting:
if $\mathcal R$ is left-linear and left-finite
then it has the compression property,
as proved in \cite{Kennaway.Klo.Sle.Vri.95}
in the topological setting.
Our proof relies on \cref{the:compression} and 
on two lemmas: the first one isolates the finite portion
of an infinite rewriting that is necessary to produce a given finite prefix
of the output (\emph{pattern extraction}),
the second one performs the remaining infinitary rewriting
on the branches starting from this prefix (\emph{pattern filling}).
These lemmas form the base case of a straightforward induction
proving the property $\property{Q}$.
A concrete illustration of the whole compression procedure
is given in \cref{app:example}.

\begin{definition}
	A term $l \in \foterms$ is \defemph{linear}
	if no variable occurs twice (or more) in ~$l$.
	A rewrite rule $(l,r)$ is
	\defemph{left-linear} (resp. \defemph{left-finite}) 
	if $l$ is linear (resp. $l \in \foterms$).
	An \itrs  $\mathcal R$ is 
	\defemph{left-linear} (resp. \defemph{left-finite}) 
	if each rule $(l,r) \in \mathcal R$ is so.
\end{definition}

\begin{lemma}[pattern extraction] \label{lem:fo-pattern-extract}
	Consider an ordinal $\delta$ such that
	$\property{P}_{\delta}$ holds, and
	$s \in \foiterms$, $\sigma : \Var \to \foiterms$
	and a linear $l \in \foterms$
	such that $s \redi[\delta] \sigma \cdot l$.
	Then there is a substitution $\tau : \Var \to \foiterms$ such that
	$s \reds \tau \cdot l$ and
	$\forall x \in \Var,\ \tau(x) \redi[\delta] \sigma(x)$.
\end{lemma}
	
	\begin{proof}
	By induction over $l$.
	If $l = x$, then we define $\tau(x) \coloneqq s$
	and $\tau(y) \coloneqq \sigma(y)$ otherwise.
	If $l = \focons c(l_1,\dots,l_k)$ then there is derivation as follows:
	\[\begin{myprooftree}
		\hypo{ s \redsequence \delta n \focons c(t_1,\dots,t_k) }
		\hypo{ t_1 \redi[\delta] \sigma \cdot l_1 }
		\hypo{ \dots }
		\hypo{ t_k \redi[\delta] \sigma \cdot l_k }
		\infer[double]3[lift_c]{ \focons c(t_1,\dots,t_k) \redig[\delta]
			\focons c(\sigma \cdot l_1,\dots,\sigma \cdot l_k) }
		\infer2[split]{ s \redi[\delta] \sigma \cdot l
			= \focons c(\sigma \cdot l_1,\dots,\sigma \cdot l_k) }
	\end{myprooftree}\]
	By $\property{P}_{\delta}$ applied to the left premise of
	\proofrule{split}, there are an ordinal $\epsilon < \delta$
	and $s_1,\dots,s_k \in \foiterms$ such that
	$s \reds \focons c(s_1,\dots,s_k) \redigs[\epsilon]
	\focons c(t_1,\dots,t_k)$.
	As a consequence, for all $1 \leq i \leq k$ there are reductions
	$s_i \redis[\epsilon] t_i \redi[\delta] \sigma \cdot l_i$,
	therefore by \cref{lem:redig-max-transitivity}
	$s_i \redi[\delta] \sigma \cdot l_i$.
	By induction there exist substitutions $\tau_i$ such that
	$s_i \reds \tau_i \cdot l_i$ and
	$\forall x \in \Var,\ \tau_i(x) \redi[\delta] \sigma(x)$.
	By linearity of $l$, each variable occurring in $l$
	occurs in exactly one of the $l_i$
	hence we can define, for all $x \in \Var$,
	$\tau(x) \coloneqq \tau_i(x)$ if $x$ occurs in some $l_i$
	and $\tau(x) \coloneqq \sigma(x)$ otherwise.
	As a consequence, $\forall x \in \Var,\ \tau(x) \redi[\delta] \sigma(x)$.
	In addition,
	$	s \reds \focons c(s_1,\dots,s_k) \reds
		\focons c(\tau_1 \cdot l_1, \dots, \tau_k \cdot l_k)
		= \tau \cdot l$.
	\end{proof}

\begin{lemma}[pattern filling] \label{lem:fo-pattern-fill}
	For all $r \in \foiterms$ and $\sigma, \tau : \Var \to \foiterms$,
	if $\forall x \in \Var,\ \tau(x) \redi[\delta] \sigma(x)$
	then $\tau \cdot r \redi[\delta] \sigma \cdot r$.
\end{lemma}
	
	\begin{proof}
	By coinduction on $r \in \foiterms$.
	If $r$ is just a variable the result follows by the assumption.
	Otherwise $r = \focons c(s_1,\dots,s_{\arity(\focons c)})$.
	For all $1 \leq i \leq \arity(\focons c)$
	we build $\tau \cdot s_i \redi[\delta] \sigma \cdot s_i$ coinductively.
	By applying the rule \proofrule{lift_{\focons c}}
	we obtain $\tau \cdot r \redig[\delta] \sigma \cdot r$
	and we conclude by \cref{lem:redig-included-in-redi}.
	\end{proof}

\begin{theorem} \label{the:fo-Q}
	If $\mathcal R$ is left-linear and left-finite then
	$\mathord{\red}$ satisfies the property $\property Q$,
	and thus has the compression property.
\end{theorem}


	\begin{proof}[Proof sketch]
	Take an ordinal $\delta$ and terms $s,t,t' \in \foiterms$
	such that $\property{P}_{\delta}$
	and $s \redi[\delta] t \red t'$.
	We want to build a term $s' \in \foiterms$
	and reductions $s \reds s' \redi[\delta] t'$.
	By induction over $t \red t'$,
	\begin{ienumerate}
	\item if there are a rule $(l,r) \in \mathcal R$
		and a substitution $\sigma : \Var \to \foiterms$
		such that $t = \sigma \cdot l$ and $t' = \sigma \cdot r$,
		then we apply
		\cref{lem:fo-pattern-extract,lem:fo-pattern-fill}.
	\item otherwise we proceed by induction.
	\end{ienumerate}
	\end{proof}




\subsection{An example in first-order rewriting}
\label{app:example}

\NewDocumentCommand \infapp {m} { {#1}^{\omega} }

\colorlet{1}{green}
\colorlet{2}{red}
\colorlet{21}{orange}
\colorlet{22}{magenta!70}
\colorlet{3}{blue}
\colorlet{30}{cyan}

In \cref{eq:intro-example-1,eq:intro-example-2} in the introduction,
we gave the example of the \itrs defined on the signature
$\Sigma \coloneqq \set{a,f,g}$
with $\arity(a) \coloneqq 0$, $\arity(f) \coloneqq 1$ 
and $\arity(g) \coloneqq 1$,
by the following rules:
$a \red_1 f(g(a))$ and $g(f(x)) \red_2 f(x)$;
we presented a strongly converging rewriting sequence
$a \redi f(f(f(\dots)))$ of length $\omega \cdot 2$
as well as an equivalent sequence of length $\omega$.
Let us translate the initial sequence into a derivation
of $a \redi f(f(f(\dots)))$
in the coinductive formalism from \cref{def:redi-fo},
and describe how the proofs of \cref{the:compression,the:fo-Q}
compress it in a derivation of $a \redo f(f(f(\dots)))$.

We use the following notations:
$\infapp f \coloneqq f(f(f(\dots)))$,
$\infapp{fg} \coloneqq f(g(f(g(\dots))))$, and
$\infapp{gf} \coloneqq g(f(g(f(\dots))))$.
The initial sequence $a \redi \infapp f$ of length $\omega \cdot 2$,
namely \cref{eq:intro-example-1},
corresponds to the following derivation:
\begin{equation} \label{app:example:eq:1}
	\begin{myprooftree}[center]
	\hypo{ \color{1} a \red_1 f(g(a)) \hspace{-3em} }
	\hypo{ a \red_1 f(g(a)) }
	\hypo{\parbox[t][0pt]{0pt}{
		\begin{tikzpicture}[baseline,color=2]
		\draw [->,dashed] (0,0)
		to[out=90,in=180] (0.9,0.4)
		to[out=0,in=90] (1.5,-1.5)
		to[out=270,in=0] (0,-3.2);
		\end{tikzpicture}
	}}
	\infer[no rule]1{ \color{2} f(g(a)) \redig[0] \infapp{fg} }
	\infer2{ \color{2} a \redi[0] \infapp{fg} }
	\infer[double]1{ \color{2} g(a) \redig[0] \infapp{gf} }
	\rewrite{\color{2}\box\treebox }
	\infer1{ g(a) \redi[0] \infapp{gf} }
	\infer[double]1{ f(g(a)) \redig[0] \infapp{fg} }
	\hypo{ \color{3} \infapp{gf} \red_2 \infapp{fg} }
	\hypo{\parbox[t][0pt]{0pt}{
		\begin{tikzpicture}[baseline,color=30]
		\draw [->,dashed] (0,0)
		to[out=90,in=180] (0.7,0.4)
		to[out=0,in=90] (1.3,-1.3)
		to[out=270,in=0] (0,-1.8);
		\end{tikzpicture}
	}}
	\infer[no rule]1{ \color{30} \infapp{fg} \redig[1] \infapp f }
	\rewrite{\color{30}\box\treebox }
	\infer2{ \infapp{gf} \redi[1] \infapp f }
	\infer[double]1{ \infapp{fg} \redig[1] \infapp f }
	\infer[snake]3{ a \redo \infapp{f} }
	\end{myprooftree}
\end{equation}
We omit the label of the rules but they should be clear:
simple rules are occurrences of \proofrule{split}
or \proofrule{split^\omega},
double rules are occurrences of \proofrule{lift_{\focons c}}
or \proofrule{lift_{\focons c}^\omega}
for $\focons c \in \set{f,g}$,
and the wavy rule is the rule \proofrule{comp} from the proof of
\cref{the:compression}, delimiting the compressed prefix
and branches yet to be compressed.

Next, we apply transitivity (\cref{lem:redig-max-transitivity}).
The {\color{2} red}, {\color{3} blue} and {\color{30} cyan}
subderivations are reused without modification,
and the compressed prefix grows:
\begin{equation} \label{app:example:eq:2}
	\begin{myprooftree}[center]
	\hypo{ \color{1} a \red_1 f(g(a)) \hspace{-6em} }
	\hypo{ \color{21} a \red_1 f(g(a)) }
	\hypo{\parbox[t][0pt]{0pt}{
		\begin{tikzpicture}[baseline,color=22]
		\draw [->,dashed] (0,0)
		to[out=90,in=180] (0.8,0.4)
		to[out=0,in=90] (1.6,-2)
		to[out=270,in=0] (0,-3.2);
		\end{tikzpicture}
	}}
	\infer[no rule]1{ \color{22} g(a) \redig[0] \infapp{gf} }
	\infer1{ \color{22} g(a) \redi[0] \infapp{gf} }
	\infer[double]1{ \color{22} f(g(a)) \redig[0] \infapp{fg} }
	\rewrite{\color{22}\box\treebox}
	\infer2{ \color{2} a \redi[0] \infapp{fg} }
	\infer[double]1{ \color{2} g(a) \redig[0] \infapp{gf} }
	\rewrite{\color{2}\box\treebox }
	\hypo{ \color{3} \infapp{gf} \red_2 \infapp{fg} }
	\hypo{ \vdots }
	\infer[double]1{ \color{30} \infapp{fg} \redig[1] \infapp f }
	\rewrite{\color{30}\box\treebox }
	\infer[snake]3{ g(a) \redo \infapp f }
	\infer[double]1{ f(g(a)) \redog \infapp f }
	\infer2{ a \redo \infapp{f} }
	\end{myprooftree}
\end{equation}

Now we need to permute $\color{2} g(a) \redig[0] \infapp{gf}$
and $\color{3} \infapp{gf} \red_2 \infapp{fg}$.
To do so, we do \emph{pattern extraction} (\cref{lem:fo-pattern-extract}):
we start with the necessary steps of the former
(in fact only the first step $\color{21} a \red_1 f(g(a))$,
performed in the context $g(-)$)
so that the latter can be performed.
Then we do \emph{pattern filling} (\cref{lem:fo-pattern-fill}),
performing {\color{22} the remainder} of the former rewriting.
Concretely:
\begin{equation} \label{app:example:eq:3}
	\begin{myprooftree}[center]
	\hypo{ \color{1} a \red_1 f(g(a)) \hspace{-6em} }
	\hypo{ g({\color{21}a}) \mathrel{\color{21}\red_1}
		g({\color{21} f(g(a)) }) \mathrel{\color{3}\red_2}
		{\color{3} f(g(a)) } }
	\hypo{\vdots}
	\infer[double]1{ \color{22} f(g(a)) \redig[0] \infapp{fg} }
	\rewrite{\color{22}\box\treebox}
	\hypo{\vdots}
	\infer[double]1{ \color{30} \infapp{fg} \redig[1] \infapp f }
	\rewrite{\color{30}\box\treebox }
	\infer[snake]3{ g(a) \redo \infapp f }
	\infer[double]1{ f(g(a)) \redog \infapp f }
	\infer2{ a \redo \infapp{f} }
	\end{myprooftree}
\end{equation}

Finally, observe that the {\color{22} pink} and {\color{30} cyan}
subderivations are exactly the ones we merged between
\cref{app:example:eq:1,app:example:eq:2},
hence we can apply again the transitivity \cref{lem:redig-max-transitivity}:
\begin{equation} \label{app:example:eq:4}
	\begin{myprooftree}[center]
	\hypo{ a \red_1 f(g(a)) \hspace{-5em} }
	\hypo{ g(a) \red_{1,2}^* f(g(a)) \hspace{-5em} }
	\hypo{\vdots}
	\infer[double]1{ \color{2} g(a) \redig[0] \infapp{gf} }
	\rewrite{\color{2}\box\treebox }
	\hypo{ \color{3} \infapp{gf} \red_2 \infapp{fg} }
	\hypo{ \vdots }
	\infer[double]1{ \color{30} \infapp{fg} \redig[1] \infapp f }
	\rewrite{\color{30}\box\treebox }
	\infer[snake]3{ g(a) \redo \infapp f }
	\infer[double]1{ f(g(a)) \redog \infapp f }
	\infer2{ g(a) \redo \infapp f }
	\infer[double]1{ f(g(a)) \redog \infapp f }
	\infer2{ a \redo \infapp{f} }
	\end{myprooftree}
\end{equation}
We obtain again a wider compressed prefix.
We continue coinductively above the wavy line,
as in \cref{app:example:eq:2}.

\subsection{Compression for (coinductive) infinitary λ-calculi}
\label{ssec:compression-lambda}

The very same work can be done for the $abc$-infinitary λ-calculi
introduced in \cref{ssec:lambda}:

\begin{theorem} \label{the:lam-Q}
	The relation $\mathord{\red}$ on $\liiiterms abc$
	satisfies the property $\property Q$,
	and thus has the compression property.
\end{theorem}

The proof is detailed in \shortlongversion
{the appendices of the long version of this article \cite{Cerda.Sau.26a}}
{\cref{app:compression-lambda}}.
A noticeable consequence of the compression property in this setting
is that our coinductive presentation of $\redi$
(as in \cref{the:M-equiv-C-lambda})
is equivalent to its usual one in the literature,
where only $\redo$ is defined
\cite{Endrullis.Pol.13,Cerda.24}.

\section{Compressing \texorpdfstring{$\muMALLi$}{μMALL∞} cut-elimination 
sequences}
\label{sec:compressing-mumall}

In this last section, we prove compression for an example of
the other kind of infinitary rewriting appearing in the literature:
infinitary cut-elimination in non-wellfounded proof systems.
We focus on the multiplicative-additive linear logic
with fixed points $\muMALL$
and the non-wellfounded proof system $\muMALLi$ for this logic,
as a compression lemma for infinitary cut-elimination in this system
can crucially be used to prove cut-elimination
for a whole range of non-wellfounded proof systems for
(intuitionistic, classical, linear) logics with fixed points
\cite[Prop.~26]{Saurin.23}.

\subsection{The non-wellfounded proof system 
\texorpdfstring{$\muMALLi$}{μMALL∞}}
\label{ssec:muMALLi}

We recall the definitions of the formulæ of $\muMALL$
and the rules of $\muMALLi$,
from which we build a set $\dtrees[\muMALLi]$ of
non-wellfounded many-sorted terms.
These trees are usually called \emph{pre-proofs} in the literature,
where \defemph{proofs} are then the subset of pre-proofs
satisfying some validity criterion.
However compression and validity are quite orthogonal concerns, as
\begin{ienumerate}
\item the following development can be performed entirely
	at the more general level of pre-proofs,
\item the validity of the limit of a compressed cut-elimination sequence 
	is not ensured by the compression lemma,
	but by the fact that the original transfinite sequence
	has a limit that is a valid derivation;
\end{ienumerate}
hence there is no need
to discuss or even define validity in this work.

We closely follow the exposition from \cite{Saurin.23a},
making some technicalities more precise
and showing how this material
can be encoded as an instance of the general framework
we introduced in \cref{sec:infinintary-rewriting}.

	Fix a set $\mathcal A$ of \defemph{atomic formulæ} and
	a countable set $\mathcal X$ of \defemph{fixed point variables}.
	The set $\mathcal F_0$ of $\muMALL$ \defemph{pre-formulæ}
	is defined by induction by:
	\begin{align*}
		\mathcal F_0 \quad \ni \quad F,G,\dots \quad \coloneqq \quad  &
		\hspace{\widthof{|\ }} A \ |\ \llneg A 
			\ |\ 0 \ |\ 1 \ |\ \top \ |\ \bot \ |\  F \llpar G \ |\ 
			F \lltens G \ |\ F \llplus G \ |\ F \llwith G \\&
		|\ X \ |\ \mu X.F \ |\ \nu X.F
			& \quad\mathllap{(A \in \mathcal A, X \in \mathcal X).}
	\end{align*}
	The set $\mathcal F$ of $\muMALL$ \defemph{formulæ}
	is the set of all pre-formulæ containing no free fixed point variable
	(\ie each $X \in \mathcal X$ only occurs in the scope
	of fixed point constructors $\mu X$ and $\nu X$).
	The set $\mathcal S$ of $\muMALL$ \defemph{sequents}
	is the set of all finite lists of formulæ,
	denoted by $\vdash F_1,\dots,F_n$.
	As usual, portions of sequents are denoted by $\Gamma$, $\Delta$, etc.
%
	Negation is the involution $\llneg{(-)} : \mathcal F \to \mathcal F$ 
	inductively defined by:
	$\llneg{(\llneg A)} \coloneqq A$,
	$\llneg 0 \coloneqq \top$, 
	$\llneg 1 \coloneqq \bot$,
	$\llneg{(F \llpar G)} \coloneqq \llneg F \lltens \llneg G$,
	$\llneg{(F \llplus G)} \coloneqq \llneg F \llwith \llneg G$ and
	$\llneg{(\mu X.F)} \coloneqq \nu X.\llneg F$.

\begin{figure}
	\begin{myprooftree}
	\infer[double]0[ax_{\mathnormal F}]{ \vdash F, \llneg F }
	\end{myprooftree}
\qquad
	\begin{myprooftree}
	\hypo{ \vdash \Gamma, F }
	\hypo{ \vdash \Delta, \llneg F }
	\infer[double]2[cut]{ \vdash \Gamma, \Delta }
	\end{myprooftree}
\qquad
	\begin{myprooftree}
	\hypo{ \vdash F_{1,1}, \dots, F_{1,n_1} }
	\hypo{ \dots }
	\hypo{ \vdash F_{k,1}, \dots, F_{k,n_k} }
	\infer[double]3[mcut_{\mathnormal{k,\vec n,\mathord{\cutrel}}}]
		{ \vdash \Gamma }
	\end{myprooftree}
\\[\topsep]
	\begin{myprooftree}
	\hypo{ \vdash F_{\sigma(1)}, \dots, F_{\sigma(n)} }
	\infer[double]1[x_{\mathnormal \sigma}]{ \vdash F_1,\dots,F_n }
	\end{myprooftree}
\qquad
	\begin{myprooftree}
	\infer[double]0[1]{ \vdash 1 }
	\end{myprooftree}
\qquad
	\begin{myprooftree}
	\infer[double]0[\top_{\!\mathnormal \Gamma}]{ \vdash \Gamma, \top }
	\end{myprooftree}
\qquad
	\begin{myprooftree}
	\hypo{ \vdash \Gamma }
	\infer[double]1[\bot]{ \vdash \Gamma, \bot }
	\end{myprooftree}
\qquad
	\begin{myprooftree}
	\hypo{ \vdash \Gamma, F, G }
	\infer[double]1[\llpar]{ \vdash \Gamma, F \llpar G }
	\end{myprooftree}
\qquad
	\begin{myprooftree}
	\hypo{ \vdash \Gamma, F }
	\hypo{ \vdash \Delta, G }
	\infer[double]2[\lltens]{ \vdash \Gamma, \Delta, F \lltens G }
	\end{myprooftree}
\\[\topsep]
	\begin{myprooftree}
	\hypo{ \vdash \Gamma, F_i }
	\infer[double]1[\llplus_{\mathnormal{i, F_{1-i}}}]
		{ \vdash \Gamma, F_0 \llplus F_1 }
	\end{myprooftree}
\qquad
	\begin{myprooftree}
	\hypo{ \vdash \Gamma, F }
	\hypo{ \vdash \Gamma, G }
	\infer[double]2[\llwith]{ \vdash \Gamma, F \llwith G }
	\end{myprooftree}
\qquad
	\begin{myprooftree}
	\hypo{ \vdash \Gamma, F[\mu X.F/X] }
	\infer[double]1[\mu]{ \vdash \Gamma, \mu X.F }
	\end{myprooftree}
\qquad
	\begin{myprooftree}
	\hypo{ \vdash \Gamma, F[\nu X.F/X] }
	\infer[double]1[\nu]{ \vdash \Gamma, \nu X.F }
	\end{myprooftree}
\caption{The derivation rules of $\muMALLi$.
	$F[G/X]$ denotes the formula obtained by substituting the fixed point
	variable $X$ with $G$ in $F$, in a capture-avoiding manner.
	}
\label{fig:muMALLi}
\end{figure}

We present the derivation rules of the system $\muMALLi$
in \cref{fig:muMALLi}.
Notice that 
\begin{ienumerate}
\item instead of the usual binary exchange rule,
we define a family of rules \proofrule{x_{\mathnormal \sigma}}
parametrised by a permutation $\sigma : [1,n] \to [1,n]$,
and acting on sequents of length $n$;
\item for the rules \proofrule{ax}, \proofrule{\top} and \proofrule{\llplus}
to be functional, we also need to present them as families of rules
parametrised respectively by $F \in \mathcal F$, $\Gamma \in \mathcal S$,
and $(i, F_{i-1}) \in \Bool \times \mathcal F$;
\item the system features an $n$-ary cut rule,
	the \defemph{multicut rule} \cite{Fortier.San.13,Baelde.Dou.Sau.16}.
	Indeed, cut-elimination theorems for $\muMALLi$ rely on
	moving finitely many consecutive cuts upwards
	(and not of single cuts): a multicut represents such a tree prefix.
	The multicut rule is parametrised by
	$k \in \Nat$, $\vec n \coloneqq (n_1, \dots, n_k) \in \Nat^k$
	and a relation $\mathord{\cutrel} \subset (\Nat^2)^2$
	(a symbol that is reminiscent of cut links in proof nets)
	relating dual formulæ $F$ and $\llneg F$
	that are cut against each other,
	under some conditions on $\cutrel$ ensuring that this is done sensibly,
	\eg no formula is cut against two different formulæ;
	the conclusion of the rule contains all the formulæ that
	have not been cut.
	A formal definition is given in \shortlongversion
	{the appendices of the long version of this article \cite{Cerda.Sau.26a}}
	{\cref{app:multicut}}.
\end{ienumerate}

	$\muMALLi$ \defemph{pre-proofs} are 
	the elements of $\dtrees[\muMALLi]$,
	the many-sorted terms obtained by the construction rules
	from \cref{fig:muMALLi}%
	.

\subsection{Compression of infinitary cut-elimination}

As explained in the introduction of this article,
cut-elimination theorems for non-wellfounded proof systems
rely on a rewriting relation \enquote{moving the cuts upwards}
so that the corresponding infinitary rewriting sequences
produce cut-free derivations.
For $\muMALLi$ this rewriting relation is defined by three kinds
of root rewriting steps:
\begin{enumerate}
\item steps handling technicalities related to multicuts, 
	\eg merging a cut into a multicut,
\item \defemph{principal steps}, corresponding to the situation where
	dual formulæ are (multi)cut against each other, for example:
	\[
	\begin{myprooftree}
		\hypo{\moresequents}
		\hypo{\vdash \Gamma, F_0}
		\hypo{\vdash \Gamma,F_1}
		\infer[double]2[\llwith]{\vdash \Gamma,F_0\llwith F_1}
		\hypo{\vdash \Delta,F_i^\bot}
		\infer[double]1[\llplus_{\mathnormal{i, F_{i-1}}}]
			{\vdash \Delta, \llneg{F_0} \llplus \llneg{F_1} } 
		\infer[double]3[mcut_{\mathnormal{k+2, \vec n, \cutrel}}]
			{\vdash \conclusion}
		\end{myprooftree}
	\widered
		\begin{myprooftree}
		\hypo{\moresequents}
		\hypo{\vdash \Gamma,F_i}
		\hypo{\vdash \Delta,\llneg{F_i}}
		\infer[double]3[mcut_{\mathnormal{k+2, \vec n, \cutrel}}]
			{\vdash \conclusion}
		\end{myprooftree}
	\]
\item \defemph{commutative steps}, corresponding to the situation where
	a multicut is permuted with the last rule of one of its premises,
	for example:
	\[
	\begin{myprooftree}
		\hypo{\moresequents}
		\hypo{\vdash \Gamma,F,G}
		\infer[double]1[\llpar]{\vdash \Gamma, F\llpar G} 
		\infer[double]2[mcut_{\mathnormal{k+1,\vec n,\cutrel}}]
			{\vdash \conclusion, F\llpar G}
		\end{myprooftree}
	\widered
		\begin{myprooftree}
		\hypo{\moresequents}
		\hypo{\vdash \Gamma,F,G}
		\infer[double]2[mcut_{\mathnormal{k+1,\vec n',\cutrel}} 
		]{\vdash \conclusion, F,G}
		\infer[double]1[\llpar]{\vdash \conclusion, F\llpar G}
		\end{myprooftree}
	\]
\end{enumerate}
Because of space constraints, 
the description of all root steps
is given in \shortlongversion
{the appendices of the long version of this article \cite{Cerda.Sau.26a}}
{\cref{app:elimination}}.

\begin{definition}
	\defemph{Cut-elimination} is the relation $\red$ on $\dtrees[\muMALLi]$
	generated by the root steps given in \shortlongversion
	{the appendices of the long version of this article \cite{Cerda.Sau.26a}}
	{\cref{app:elimination}},
	using the construction from \cref{def:red}.
\end{definition}

In the remainder of this section, we prove that compression holds
for $\muMALLi$ cut-elimination.
Although the definitions for this case of infinitary rewriting
are way heavier than for first-order rewriting or the λ-calculus,
the proof of compression happens to be simpler than in those two frameworks thanks to the following observation:
all root steps defining $\red$ rewrite only \emph{finite (linear) patterns}
into \emph{finite patterns}: it is obvious, for instance,
for the few examples given above.
The proof follows the same steps as the two previous ones.

\begin{definition}
	Given a family $\mathcal D$ of construction rules
	as in \cref{def:dtree}, \defemph{derivation patterns}
	are defined as follows:
	\begin{ienumerate}
	\item for all $\drule r \in \mathcal D$ of arity $k$, 
		$\drule r(\hole_1,\dots,\hole_k)$
		is a derivation pattern,
	\item for all $\drule r \in \mathcal D$ of arity $k$,
		for all derivation patterns $p_i(\hole_{i,1},\dots,\hole_{i,l_i})$
		for $1 \leq i \leq k$,
		$\drule r( p_1(\hole_{1,1},\dots,\hole_{1,l_1}), \dots,
		p_k(\hole_{k,1},\dots,\hole_{k,l_k}) )$ is a derivation pattern.
	\end{ienumerate}
	We denote by $p(s_1,\dots,s_k)$ the many-sorted term obtained by 
	substituting a term $s_i$ to each symbol $\hole_i$ 
	in the derivation pattern $p(\hole_1,\dots,\hole_k)$.
\end{definition}

\begin{lemma}[pattern extraction] \label{lem:muMALL-pattern-extract}
	Consider an ordinal $\delta$ such that $\property{P}_\delta$ holds,
	a derivation pattern $p(\hole_1, \dots, \hole_k)$ and
	$s,t_1,\dots,t_k \in \dtrees[\muMALLi]$ such that
	$s \redi[\delta] p(t_1,\dots,t_k)$.
	Then there exist $s'_1,\dots,s'_k \in \dtrees[\muMALLi]$ such that
	$s \reds p(s'_1,\dots,s'_k)$ and 
	for all $i \in [1,k]$, $s'_i \redi[\delta] t_i$.
\end{lemma}
	
	\begin{proof}
	The proof is identical as the one for
	\cref{lem:fo-pattern-extract}:
	we proceed by a straightforward induction over $p$,
	using $\property{P}_\delta$ and
	\cref{lem:P-implies-preponement-root-steps}.
	\end{proof}

\begin{lemma}[pattern filling] \label{lem:muMALL-pattern-fill}
	For all derivation pattern $q(\hole_1, \dots, \hole_k)$ and
	for all $s'_1,\dots,s'_k$, $t_1,\dots,t_k \in \dtrees[\muMALLi]$,
	if for all $i \in [1,k]$, $s'_i \redi[\delta] t_i$
	then $q(s'_1,\dots,s'_k) \redi[\delta] q(t_1,\dots,t_k)$.
\end{lemma}
	
	\begin{proof}
	Again, the proof is similar to the one for \cref{lem:fo-pattern-fill},
	but is in fact simpler: as we only consider a finite pattern~$q$,
	we just need to proceed by induction on~$q$.
	\end{proof}

\begin{theorem}
	$\red$ satisfies the property $\property{Q}$,
	and thus has the compression property.
\end{theorem}
	
	\begin{proof}
	By induction on arbitrary cut-elimination steps $s \red t$
	(as defined in \cref{def:red}), one can show that any such step
	has the shape $p(s_1,\dots,s_k) \red q(s_1,\dots,s_k)$
	for some derivation prefixes $p$ and $q$
	(to be precise, $q$ may not use some of the $s_i$
	but this has no incidence on the proof).
	Indeed, as explained above it is the case for all
	cut-elimination root steps
	(as defined in \shortlongversion
	{the appendices of the long version of this article \cite{Cerda.Sau.26a}}
	{\cref{app:elimination}}),
	and the induction case is trivial.
	
	As a consequence, if we take an ordinal $\delta$ and 
	$s,t,t' \in \dtrees[\muMALLi]$
	such that $\property{P}_{\delta}$
	and $s \redi[\delta] t \red t'$,
	we can write $t = p(t_1,\dots,t_n)$ and $t' = q(t_1,\dots,t_n)$.
	By \cref{lem:muMALL-pattern-extract,lem:muMALL-pattern-fill}
	we obtain
	$s \reds p(s'_1,\dots,s'_k) \redi[\delta] q(t_1,\dots,t_n) = t'$.
	\end{proof}


\section{Conclusion and further work}
\label{sec:conclusion}

In the present paper, we did the following:
\begin{enumerate}
\item introduce a generic framework for rewriting
	infinitary objects,
	and extend to this setting
	the equivalence between topologically and coinductively
	presented infinitary rewriting,
\item perform the first generic treatment of the compression property
	for coinductive infinitary rewriting,
	by providing a novel characterisation that can be readily applied
	to all the existing infinitary term rewriting systems
	(\eg,
	first-order rewriting and the λ-calculus),
\item take advantage of the previous work to present
	cut-elimination in a non-wellfounded proof system
	as a coinductive rewriting procedure,
	and fully work out the example of the system $\muMALLi$
	for which we provide a simple proof of compression
	(a property whose significance in this setting
	has been recalled in the introduction).
\end{enumerate}

The former two achievements unify and complete several threads
in the literature and are meant to provide a reasonable level of generality
for future investigations into coinductive
rewriting.
Regarding compression, a natural follow-up would be to investigate
the effectiveness of the procedure introduced in the proof of
\cref{the:compression}:
given a computable derivation witnessing a rewriting $s \redi t$,
is it computable to compress it into $s \redo t$,
as can be conjectured from our proof?
The suitable notion of \enquote{computable}
remains to be clarified.

The latter achievement is also a first step towards developing
non-wellfounded proof theory
(especially of the system $\muLLi$ for linear logic with fixed points)
in a coinductive setting, which is our longer term goal:
a result we typically aim at is a fully coinductive proof of 
cut-elimination for $\muLLi$.
Similar efforts are currently being conducted in this direction,
\eg by Sierra-Miranda \emph{et al.} \cite{Sierra-Miranda.Stu.24}
who present non-wellfounded proofs
for Gödel-Löb logic and Grzegorczyk modal logic in a coalgebraic way.
A benefit of such a coinductive treatment would be to help in 
the formalisation of the meta-theory of these systems
in proof assistants such as Rocq. 

\clearpage
\bibliography{compression.bib}

\begin{thebibliography}{10}

\bibitem{Afs.Klo.25}
Bahareh Afshari and Johannes Kloibhofer.
\newblock Cut elimination for cyclic proofs: {A} case study in temporal logic.
\newblock In Alexis Saurin, editor, {\em Proceedings Twelfth Workshop on Fixed
  Points in Computer Science (FICS 2024)}, number 435 in EPTCS, pages 21--40,
  2025.
\newblock \href {https://doi.org/10.4204/EPTCS.435.3}
  {\path{doi:10.4204/EPTCS.435.3}}.

\bibitem{Appel.Mel.Ric.Vou.07}
Andrew~W. Appel, Paul-André Melliès, Christopher~D. Richards, and Jérôme
  Vouillon.
\newblock A very modal model of a modern, major, general type system.
\newblock {\em {ACM} {SIGPLAN} Notices}, 42(1):109--122, 2007.
\newblock \href {https://doi.org/10.1145/1190215.1190235}
  {\path{doi:10.1145/1190215.1190235}}.

\bibitem{Baelde.Dou.Sau.16}
David Baelde, Amina Doumane, and Alexis Saurin.
\newblock Infinitary proof theory: the multiplicative additive case.
\newblock In {\em 25th EACSL Annual Conference on Computer Science Logic (CSL
  2016)}. 2016.
\newblock \href {https://doi.org/10.4230/LIPICS.CSL.2016.42}
  {\path{doi:10.4230/LIPICS.CSL.2016.42}}.

\bibitem{Bahr.10}
Patrick Bahr.
\newblock Partial order infinitary term rewriting and böhm trees.
\newblock In {\em Proceedings of the 21st International Conference on Rewriting
  Techniques and Application}, 2010.
\newblock \href {https://doi.org/10.4230/LIPICS.RTA.2010.67}
  {\path{doi:10.4230/LIPICS.RTA.2010.67}}.

\bibitem{Barendregt.77}
Henk~P. Barendregt.
\newblock The type free lambda calculus.
\newblock In Jon Barwise, editor, {\em Handbook of Mathematical Logic}, pages
  1091--1132. Elsevier, 1977.
\newblock \href {https://doi.org/10.1016/s0049-237x(08)71129-7}
  {\path{doi:10.1016/s0049-237x(08)71129-7}}.

\bibitem{Barendregt}
Henk~P. Barendregt.
\newblock {\em The Lambda Calculus}.
\newblock Elsevier, Amsterdam, 2 edition, 1984.

\bibitem{Barr.93}
Michael Barr.
\newblock Terminal coalgebras in well-founded set theory.
\newblock {\em Theoretical Computer Science}, 114(2):299--315, 1993.
\newblock \href {https://doi.org/10.1016/0304-3975(93)90076-6}
  {\path{doi:10.1016/0304-3975(93)90076-6}}.

\bibitem{Cerda.24}
Rémy Cerda.
\newblock {\em Taylor Approximation and Infinitary λ-Calculi}.
\newblock Theses, Aix-Marseille Université, 2024.
\newblock URL: \url{https://hal.science/tel-04664728}.

\bibitem{Cerda.25}
Rémy Cerda.
\newblock Nominal algebraic-coalgebraic data types, with applications to
  infinitary λ-calculi.
\newblock In Alexis Saurin, editor, {\em Proceedings Twelfth Workshop on Fixed
  Points in Computer Science (FICS 2024)}, number 435 in EPTCS, 2025.
\newblock \href {https://doi.org/10.4204/EPTCS.435.5}
  {\path{doi:10.4204/EPTCS.435.5}}.

\bibitem{Cerda.Sau.26}
Rémy Cerda and Alexis Saurin.
\newblock Compression for coinductive rewritingand the cut-elimination of
  non-wellfounded proofs.
\newblock In {\em 51th International Symposium on Mathematical Foundations of
  Computer Science (MFCS 2026)}, 2026.
\newblock \href {https://doi.org/10.4230/LIPIcs.MFCS.2026.21}
  {\path{doi:10.4230/LIPIcs.MFCS.2026.21}}.

\bibitem{Czajka.20}
Łukasz Czajka.
\newblock A new coinductive confluence proof for infinitary lambda calculus.
\newblock {\em Logical Methods in Computer Science}, 16(1), 2020.
\newblock \href {https://doi.org/10.23638/LMCS-16(1:31)2020}
  {\path{doi:10.23638/LMCS-16(1:31)2020}}.

\bibitem{Das.Pou.19}
Anupam Das and Damien Pous.
\newblock Non-wellfounded proof theory for (kleene+action) (algebras+lattices).
\newblock In {\em 27th EACSL Annual Conference on Computer Science Logic (CSL
  2018)}, 2018.
\newblock \href {https://doi.org/10.4230/LIPIcs.CSL.2018.19}
  {\path{doi:10.4230/LIPIcs.CSL.2018.19}}.

\bibitem{Dershowitz.Kap.Pla.91}
Nachum Dershowitz, St{\'{e}}phane Kaplan, and David~A. Plaisted.
\newblock Rewrite, rewrite, rewrite, rewrite, rewrite, {\ldots}.
\newblock {\em Theoretical Computer Science}, 83(1):71--96, 1991.
\newblock \href {https://doi.org/10.1016/0304-3975(91)90040-9}
  {\path{doi:10.1016/0304-3975(91)90040-9}}.

\bibitem{Endrullis.Han.Hen.Pol.Sil.18}
Jörg Endrullis, Helle~Hvid Hansen, Dimitri Hendriks, Andrew Polonsky, and
  Alexandra Silva.
\newblock Coinductive foundations of infinitary rewriting and infinitary
  equational logic.
\newblock {\em Logical Methods in Computer Science}, 14(1):1860--5974, 2018.
\newblock \href {https://doi.org/10.23638/LMCS-14(1:3)2018}
  {\path{doi:10.23638/LMCS-14(1:3)2018}}.

\bibitem{Endrullis.Pol.13}
Jörg Endrullis and Andrew Polonsky.
\newblock Infinitary rewriting coinductively.
\newblock In {\em 18th~International Workshop on Types for Proofs and Programs
  (TYPES 2011)}, pages 16--27, 2013.
\newblock \href {https://doi.org/10.4230/LIPIcs.TYPES.2011.16}
  {\path{doi:10.4230/LIPIcs.TYPES.2011.16}}.

\bibitem{Fortier.San.13}
Jérôme Fortier and Luigi Santocanale.
\newblock Cuts for circular proofs: semantics and cut-elimination.
\newblock In {\em Computer Science Logic 2013 (CSL 2013)}, 2013.
\newblock \href {https://doi.org/10.4230/LIPICS.CSL.2013.248}
  {\path{doi:10.4230/LIPICS.CSL.2013.248}}.

\bibitem{Joachimski04}
Felix Joachimski.
\newblock Confluence of the coinductive λ-calculus.
\newblock {\em Theoretical Computer Science}, 311(1-3):105--119, 2004.
\newblock \href {https://doi.org/10.1016/s0304-3975(03)00324-4}
  {\path{doi:10.1016/s0304-3975(03)00324-4}}.

\bibitem{Kennaway.Klo.Sle.Vri.95}
Richard Kennaway, Jan~Willem Klop, Ronan Sleep, and Fer-Jan de~Vries.
\newblock Transfinite reductions in orthogonal term rewriting systems.
\newblock {\em Information and Computation}, 119(1):18--38, 1995.
\newblock \href {https://doi.org/10.1006/inco.1995.1075}
  {\path{doi:10.1006/inco.1995.1075}}.

\bibitem{Kennaway.Klo.Sle.Vri.97}
Richard Kennaway, Jan~Willem Klop, Ronan Sleep, and Fer-Jan de~Vries.
\newblock Infinitary lambda calculus.
\newblock {\em Theoretical Computer Science}, 175(1):93--125, 1997.
\newblock \href {https://doi.org/10.1016/S0304-3975(96)00171-5}
  {\path{doi:10.1016/S0304-3975(96)00171-5}}.

\bibitem{Ketema.12}
Jeroen Ketema.
\newblock Reinterpreting compression in infinitary rewriting.
\newblock In {\em 23rd International Conference on Rewriting Techniques and
  Applications (RTA 2012)}, 2012.
\newblock \href {https://doi.org/10.4230/LIPICS.RTA.2012.209}
  {\path{doi:10.4230/LIPICS.RTA.2012.209}}.

\bibitem{Ketema.Sim.11}
Jeroen Ketema and Jakob~Grue Simonsen.
\newblock Infinitary combinatory reduction systems.
\newblock {\em Information and Computation}, 209(6):893--926, 2011.
\newblock \href {https://doi.org/10.1016/j.ic.2011.01.007}
  {\path{doi:10.1016/j.ic.2011.01.007}}.

\bibitem{Kurz.Pet.Sev.Vri.13}
Alexander Kurz, Daniela Petrişan, Paula Severi, and Fer-Jan de~Vries.
\newblock Nominal coalgebraic data types with applications to lambda calculus.
\newblock {\em Logical Methods in Computer Science}, 9(4), 2013.
\newblock \href {https://doi.org/10.2168/lmcs-9(4:20)2013}
  {\path{doi:10.2168/lmcs-9(4:20)2013}}.

\bibitem{Lombardi.Rio.Vri.14}
Carlos Lombardi, Alejandro Ríos, and Roel de~Vrijer.
\newblock Proof terms for infinitary rewriting.
\newblock In {\em Rewriting and Typed Lambda Calculi (RTA 2014, TLCA 2014)},
  pages 303--318, 2014.
\newblock \href {https://doi.org/10.1007/978-3-319-08918-8_21}
  {\path{doi:10.1007/978-3-319-08918-8_21}}.

\bibitem{Nakano.00}
Hiroshi Nakano.
\newblock A modality for recursion.
\newblock In {\em Proceedings of the 15th~Annual {IEEE} Symposium on Logic in
  Computer Science}, 2000.
\newblock \href {https://doi.org/10.1109/lics.2000.855774}
  {\path{doi:10.1109/lics.2000.855774}}.

\bibitem{Santocanale.02}
Luigi Santocanale.
\newblock A calculus of circular proofs and its categorical semantics.
\newblock In {\em Foundations of Software Science and Computation Structures
  (FoSSaCS 2002)}, pages 357--371, 2002.
\newblock \href {https://doi.org/10.1007/3-540-45931-6_25}
  {\path{doi:10.1007/3-540-45931-6_25}}.

\bibitem{Saurin.23}
Alexis Saurin.
\newblock A linear perspective on cut-elimination for non-wellfounded sequent
  calculi with least and greatest fixed-points.
\newblock In {\em Automated Reasoning with Analytic Tableaux and Related
  Methods (TABLEAUX 2023)}, pages 203--222, 2023.
\newblock \href {https://doi.org/10.1007/978-3-031-43513-3_12}
  {\path{doi:10.1007/978-3-031-43513-3_12}}.

\bibitem{Saurin.23a}
Alexis Saurin.
\newblock A linear perspective on cut-elimination for non-wellfounded sequent
  calculi with least and greatest fixed-points (extended version), 2023.
\newblock Long version of \cite{Saurin.23}.
\newblock URL: \url{https://hal.science/hal-04169137}.

\bibitem{Sierra-Miranda.Stu.24}
Borja Sierra{-}Miranda, Thomas Studer, and Lukas Zenger.
\newblock Coalgebraic proof translations for non-wellfounded proofs.
\newblock In Agata Ciabattoni, David Gabelaia, and Igor Sedl{\'{a}}r, editors,
  {\em Advances in Modal Logic, AiML 2024, Prague, Czech Republic, August
  19-23, 2024}, pages 527--548. College Publications, 2024.

\bibitem{Sprenger.Dam.03}
Christoph Sprenger and Mads Dam.
\newblock On the structure of inductive reasoning: Circular and tree-shaped
  proofs in the μ-calculus.
\newblock In {\em Foundations of Software Science and Computation Structures
  (FoSSaCS 2003)}, pages 425--440, 2003.
\newblock \href {https://doi.org/10.1007/3-540-36576-1_27}
  {\path{doi:10.1007/3-540-36576-1_27}}.

\bibitem{Terese}
Terese.
\newblock {\em Term Rewriting Systems}.
\newblock Cambridge University Press, 2003.

\bibitem{Tsukada.Unn.22}
Takeshi Tsukada and Hiroshi Unno.
\newblock Software model-checking as cyclic-proof search.
\newblock {\em Proceedings of the ACM on Programming Languages}, 6(POPL):1--29,
  2022.
\newblock \href {https://doi.org/10.1145/3498725} {\path{doi:10.1145/3498725}}.

\end{thebibliography}

\shortlongversion{}{%
	\clearpage
	\appendix
	\newcounter{tmp}

\section{Compression for (coinductive) infinitary λ-calculi}
\label[appendix]{app:compression-lambda}

\begin{lemma}[pattern extraction] \label{lem:lam-pattern-extract}
	Consider an ordinal $\delta$ such that
	$\property{P}_{\delta}$ holds, and
	$s, u, v \in \liiiterms abc$
	such that $s \redi[\delta] (λx.u)v$.
	Then there are $u', v' \in \liiiterms abc$ such that
	$s \reds (λx.u')v'$
	with $u' \redi[\delta] u$ and $v' \redig[\delta] v$.
\end{lemma}
	
	\begin{proof}
	From the hypothesis $s \redi[\delta] (λx.u)v$
	we reconstruct the following premises:
	\[\begin{prooftree}
		\hypo{ s \reds tv'' }
		\hypo{ t \reds λx.u' }
		\hypo[ordoubleprem]{ u' \redi[\delta] u }
		\infer1[lift_λ]{ λx.u' \redig[\delta] λx.u }
		\infer[snake,ordoubleprem]2{ t \redi[\delta] λx.u }
		\hypo{ v'' \reds v' }
		\hypo{ v' \redig[\delta] v }
		\infer[snake,ordoubleprem]2{ v'' \redi[\delta] v }
		\infer2[lift_@]{ tv'' \redig[\delta] (λx.u)v }
		\infer[snake]2{ s \redi[\delta] (λx.u)v }
	\end{prooftree}\]
	where the wavy lines denote applications of 
	\cref{lem:P-implies-preponement-root-steps}
	thanks to $\property P_\delta$.
	The result follows.
	\end{proof}

\begin{lemma}[pattern filling] \label{lem:lam-pattern-fill}
	For all $u,u',v,v' \in \liiiterms abc$,
	if $u' \redi[\delta] u$ and $v' \redig[\delta] v$
	then $u'[v'/x] \redi[\delta] u[v/x]$.
\end{lemma}
	
	\begin{proof}
	We proceed by nested induction and coinduction
	(depending on the booleans $a,b,c$)
	over $u' \redi[\delta] u$.
	By the rule \proofrule{split}, there must be reductions
	$u' \redsequence \delta n u'' \redig[\delta] u$.
	\begin{itemize}
	\item If $u'' = x \redig[\delta] x = u$, then
		$u''[v'/x] = v' \redig[\delta] v = u[v/x]$.
	\item If $u'' = y \redig[\delta] y = u$, then
		$u''[v'/x] = y \redig[\delta] y = u[v/x]$.
	\item If $u'' = λy.w'' \redig[\delta] λy.w = u$
		with $w'' \redi[\delta] w$,
		then by induction (if $a = 0$) or coinduction (if $a = 1$)
		$w''[v'/x] \redig[\delta] w[v/x]$,
		hence by the rule \proofrule{lift_λ}
		we obtain $u''[v'/x] \redig[\delta] u[v/x]$.
	\item In the application case we proceed similarly
		and obtain $u''[v'/x] \redig[\delta] u[v/x]$.
	\end{itemize}
	In addition, it easy to verify that
	for all reduction $s \redig[\epsilon] t$
	there is a reduction $s[v'/x] \redig[\epsilon] t[v'/x]$.
	As a consequence, $u' \redsequence \delta n u''$
	entails that $u'[v'/x] \redsequence \delta n u''[v'/x]$.
	Thus we have built the two premises allowing
	to conclude by the rule \proofrule{split}.
	\end{proof}

\begin{theorem}
	The relation $\mathord{\red}$ on $\liiiterms abc$
	satisfies the property $\property Q$,
	and thus has the compression property.
\end{theorem}

	\begin{proof}
	Given an ordinal $\delta$ and terms $s,t,t' \in \liiiterms abc$
	such that $\property{P}_{\delta}$
	and $s \redi[\delta] t \red t'$,
	we want to build a term $s' \in \liiiterms abc$
	and reductions $s \reds s' \redi[\delta] t'$.
	The proof is a straightforward induction over $t \red t'$:
	in the base case where $t = (λx.u)v$ and $t' = u[v/x]$
	the result follows from 
	\cref{lem:lam-pattern-extract,lem:lam-pattern-fill};
	the induction cases are treated exactly as in \cref{the:fo-Q}.
	\end{proof}

\begin{remark}
	A noticeable consequence of the compression property in this setting
	is that the coinductive presentation of $\redi$
	given in \cref{the:M-equiv-C-lambda}
	is equivalent to the standard one
	\cite{Endrullis.Pol.13,Cerda.24}, namely:
	\[
		\begin{prooftree}
		\hypo{s \reds x}
		\infer1{s \redi x}
		\end{prooftree}
	\qquad
		\begin{prooftree}
		\hypo{s \reds λx.u}
		\hypo[ordoubleprem]{u \redi u' \strut}
		\infer2{s \redi λx.u'}
		\end{prooftree}
	\qquad
		\begin{prooftree}
		\hypo{s \reds uv}
		\hypo[ordoubleprem]{u \redi u' \strut}
		\hypo[ordoubleprem]{v \redi v' \strut}
		\infer3{s \redi u'v'}
		\end{prooftree}
	\]
	where the (co)inductive nature of the premises
	again depends on the booleans $a,b,c$.
	Indeed these rules are just the result of
	the rule \proofrule{split} followed by a rule
	\proofrule{lift_{\var x}}, \proofrule{lift_λ} or \proofrule{lift_@}.
\end{remark}

In fact, compression holds for a large family of
\enquote{infinitary combinatory reduction systems} (\textsc{icrs}),
\ie infinitary higher-order rewriting systems:
\begin{fact}[{\cite[Theorem 5.2]{Ketema.Sim.11}}]
\label{fact:compression-icrs}
	Every fully-extended, left-linear \textsc{icrs}
	has the compression property.
\end{fact}
Even though space constraints prevent us from developing this point
and the corresponding definitions,
let us mention that \textsc{icrs} fit perfectly in the generic
framework introduced in \cref{ssec:arbitrary-derivations}
and that $(\liiiterms abc, \red)$ is a particular case
of a fully-extended, left-linear \textsc{icrs}.
As a consequence, compression in this setting can be proved
using the characterisation from \cref{the:compression}.


\section{The multicut rule}
\label[appendix]{app:multicut}

\begin{definition} \label{def:multicut}
	The \defemph{multicut rule} parametrised by
	$k \in \Nat$, $\vec n \coloneqq (n_1, \dots, n_k) \in \Nat^k$
	and a relation $\mathord{\cutrel} \subset (\Nat^2)^2$
	is defined by:
	\[\begin{myprooftree}
		\hypo{ \vdash F_{1,1}, \dots, F_{1,n_1} }
		\hypo{ \dots }
		\hypo{ \vdash F_{k,1}, \dots, F_{k,n_k} }
		\infer[double]3[mcut_{\mathnormal{k,\vec n,\mathord{\cutrel}}}]
			{ \vdash \Gamma }
	\end{myprooftree}\]
	provided $\cutrel$ satisfies the following conditions:
	\begin{description}
	\item[Correctness.] The support of $\cutrel$, 
		\ie the indices $(i,j)$ such that there is at least one relation
		$(i,j) \cutrel (i',j')$ or $(i',j') \cutrel (i,j)$,
		are all such that $i \in [1,k]$ and $j \in [1,n_i]$,
	\item[Duality.] For all $\indi$ in the support of $\cutrel$
		there is a unique $\indj$ such that $\indi \cutrel \indj$,
		which is such that symmetrically $\indj \cutrel \indi$, 
		and such that $F_\indj = \llneg{F_\indi}$,
	\item[Connectedness and acyclicity.]
		If considered as a graph with vertices $[1,k]$,
		the first projection of $\cutrel$ is connected and acyclic;
		in other terms, \begin{ienumerate}
		\item for all $i, i' \in [1,k]$
			there exist $\indi_1 \cutrel \dots \cutrel \indi_m$
			and $j, j'$ such that $\indi_1 = (i,j)$ and $\indi_m = (i',j')$,
		\item the only cycles
			$(i,j) = \indi_1 \cutrel \dots \cutrel \indi_m = (i,j')$
			are due to the symmetry of the relation,
			\ie there is $p \in [1,m-2]$ such that $\indi_{p+2} = \indi_p$,
		\end{ienumerate}
	\end{description}
	and where $\Gamma$ contains all the formulæ $F_\indi$
	such that the pair $\indi$ is not in the support of $\cutrel$,
	listed in the lexicographic order over these pairs of indices.
\end{definition}


\section{Root steps for \texorpdfstring{$\muMALLi$}{μMALL∞} cut-elimination}
\label[appendix]{app:elimination}

As usual we distinguish the \defemph{principal} cut-elimination steps
corresponding to the situation where dual derivation rules
are cut against each other
(this is where a cut is actually \enquote{eliminated}),
and the \defemph{commutative} cut-elimination steps
where a (multi)cut is permuted with a derivation rule
appearing just above it
(here cuts are only \enquote{moved upwards}).
In addition, we start with two rules use to perform some derivation
transformations in the presence of multicuts:
the first one describes how a \proofrule{mcut}
can absorb a \proofrule{cut} when it meets one,
the second one permutes the premisses of a multicut.

We use the following notations:
in the premiss sequents playing an active role of the cut-elimination steps,
the principal formulæ are denoted by $F$, $G$
and the lists of remaining formulæ by $\Gamma$, $\Delta$, $\Epsilon$
(as in \cref{fig:muMALLi});
the other premiss sequents are denoted by $\moresequents$,
which stands for a list $\Zeta_1, \dots, \Zeta_k$;
the conclusion sequents are denoted by $\conclusion$.

Notice that what we describe are in fact sets of pairs of derivations,
\eg the \cutelimrule{cut}{mcut} root steps described
by the first yellow square below are all the pairs
\[\left(
	\drule{mcut_{\mathnormal{k+1,\vec n,\cutrel}}}
		(s_1,\dots,s_k,\drule{cut}(t,t')),
	\drule{mcut_{\mathnormal{k+2,\vec n',\cutrel'}}}
		(s_1,\dots,s_k,t,t')
\right)\]
satisfying the given conditions,
for $s_1,\dots,s_k,t,t' \in \dtrees[\muMALLi]$.

\subsection{Steps handling multicuts}
\label[appendix]{app:elimination:multicut-steps}

The merge \cutelimrule{cut}{mcut} root steps are 
all the pairs of the following shape:
\cutelimstep{%
	\begin{prooftree}
	\hypo{\moresequents}
	\hypo{\vdash \Gamma,F}
	\hypo{\vdash \Delta,\llneg F}
	\infer[double]2[cut]{\vdash \Gamma,\Delta}     
	\infer[double]2[mcut_{\mathnormal{k+1,\vec n,\cutrel}}]
		{\vdash \conclusion}
	\end{prooftree}
\widered
	\begin{prooftree}
	\hypo{\moresequents}
	\hypo{\vdash \Gamma,F}
	\hypo{\vdash \Delta,\llneg F}
	\infer[double]3[mcut_{\mathnormal{k+2,\vec n',\cutrel'}}]{
	\vdash \conclusion}
	\end{prooftree}
}
such that:
\begin{itemize}
\item $\vec n \coloneqq (n_1, \dots, n_k, 
	\seqcard{\Gamma}+\seqcard{\Delta})$ and
	$\vec n' \coloneqq (n_1, \dots, n_k, \seqcard\Gamma+1, 
	\seqcard\Delta+1)$,
\item on the right-hand side
    $(k+1,\seqcard{\Gamma}+1) \cutrel' (k+2,\seqcard{\Delta}+1)$ and
	$\indi \cutrel' \indj$ if $\prempred{\indi} \cutrel \prempred{\indj}$,
	with:
	\begin{align*}
		\prempred{i,j} & \coloneqq (i,j) & \text{for $i \leq k$} \\
		\prempred{k+1,j} & \coloneqq (k+1,j) 
			& \text{for $j \leq \seqcard{\Gamma}$} \\
		\prempred{k+2,j} & \coloneqq (k+1,\seqcard{\Gamma} + j) 
			& \text{for $j \leq \seqcard{\Delta}$} \\
		\prempred{i,j} & \coloneqq \text{undefined} & \text{otherwise}
	\end{align*}
	\ie $\cutrel'$ contains
	\begin{ienumerate}
    \item cuts between $F$ and $\llneg F$,
	\item the cuts from $\cutrel$ inside $\moresequents$ or 
		between $\moresequents$ and $\Gamma$ or $\Delta$.
	\end{ienumerate}
\end{itemize}

\medskip

\noindent
The premiss permutation root steps
parametrised by a permutation $\tau : [1,k] \to [1,k]$
are all the pairs of the following shape:
\cutelimstep{%
	\begin{prooftree}
	\hypo{\vdash \Zeta_1}
	\hypo{\dots}
	\hypo{\vdash \Zeta_k}
	\infer[double]3[mcut_{\mathnormal{k,\vec n,\cutrel}}]
		{ \vdash \conclusion_1,\dots,\conclusion_k }
	\end{prooftree}
\widered
	\begin{prooftree}
	\hypo{\vdash \Zeta_{\tau(1)}}
	\hypo{\dots}
	\hypo{\vdash \Zeta_{\tau(k)}}
	\infer[double]3[mcut_{\mathnormal{k,\vec n',\cutrel'}}]
		{ \vdash \conclusion_{\tau(1)},\dots,\conclusion_{\tau(k)} }
	\infer[double]1[x_{\mathnormal \sigma}]
		{ \vdash \conclusion_1,\dots,\conclusion_k}
	\end{prooftree}
}
such that:
\begin{itemize}
\item $\vec n \coloneqq (n_1,\dots,n_k)$ and 
	$\vec n' \coloneqq (n_{\tau(1), \dots, n_{\tau(k)}})$,
\item for all $(i,j)$ and $(i',j')$ in the support of $\cutrel$,
	$(\tau(i),j) \cutrel' (\tau(i'),j')$ \ifff $(i,j) \cutrel (i',j')$,
\item $\sigma$ is the permutation obtained by permuting the lists
	$\conclusion_i$ according to $\tau$ without modifying the
	ordering inside each of these lists.
\end{itemize}

\subsection{Principal reduction steps}

The \cutelimrule{ax}{mcut} root steps are 
all the pairs of the following shape:
\cutelimstep{%
	\begin{prooftree}
	\hypo{\moresequents}
	\infer[double]0[ax_{\mathnormal F}]{\vdash F, \llneg F}
	\hypo{\vdash \Gamma, F}    
	\infer[double]3[mcut_{\mathnormal{k+2, \vec n, \cutrel}}]
		{\vdash \conclusion}
	\end{prooftree}
\widered
	\begin{prooftree}
	\hypo{\moresequents}
	\hypo{\vdash \Gamma, F}
	\infer[double]2[mcut_{\mathnormal{k+1, \vec n', \cutrel'}}]
		{\vdash \conclusion}
	\end{prooftree}
}
such that:
\begin{itemize}
\item $\vec n \coloneqq (n_1, \dots, n_k, 2, \seqcard{\Gamma}+1)$ and
	$\vec n' \coloneqq (n_1,\dots,n_k,\seqcard{\Gamma}+1)$,
\item on the left-hand side $(k+1,2) \cutrel (k+2,\seqcard{\Gamma}+1)$,
	\ie $\llneg F$ is cut against the last $F$,
\item on the right-hand side
	$\indi \cutrel' \indj$ \ifff $\prempred{\indi} \cutrel \prempred{\indj}$,
	with:
	\begin{align*}
		\prempred{i,j} & \coloneqq (i,j) & \text{for $i \leq k$} \\
		\prempred{k+1,j} & \coloneqq (k+2,j) 
			& \text{for $j \leq \seqcard{\Gamma}$} \\
		\prempred{k+1,\seqcard{\Gamma}+1} & \coloneqq (k+1,1),
	\end{align*}
	\ie $\cutrel'$ contains
	\begin{ienumerate}
	\item the cuts from $\cutrel$ inside $\moresequents$ or 
		between $\moresequents$ and $\Gamma$,
	\item if there was a cut between a formula in $\moresequents$
		and the first $F$, then a cut between this formula and
		the remaining~$F$.
	\end{ienumerate}
	Thanks to the assumptions on $\cutrel$
	no cut is lost during the process,
	and $\cutrel'$ satisfies the same required assumptions.
\end{itemize}

\medskip

\noindent
The \cutelimrule{\lltens}{\llpar} root steps are 
all the pairs of the following shape:
\cutelimstep{%
	\begin{prooftree}
	\hypo{\moresequents}
	\hypo{\vdash \Gamma,F}
	\hypo{\vdash \Delta,G}
	\infer[double]2[\lltens]{\vdash \Gamma,\Delta,F\lltens G}
	\hypo{\vdash \Epsilon,\llneg F,\llneg G}
	\infer[double]1[\llpar]{\vdash \Epsilon,\llneg F\llpar \llneg G} 
	\infer[double]3[mcut_{\mathnormal{k+2, \vec n, \cutrel}}]
		{\vdash \conclusion}
	\end{prooftree}
\widered
	\begin{prooftree}
	\hypo{\moresequents}
	\hypo{\vdash \Gamma,F}
	\hypo{\vdash \Delta,G}
	\hypo{\vdash \Epsilon,\llneg F,\llneg G}
	\infer[double]4[mcut_{\mathnormal{k+3, \vec n', \cutrel'}}]
		{\vdash \conclusion}
	\end{prooftree}
}
such that:
\begin{itemize}
\item $\vec n \coloneqq (n_1, \dots, n_k,
	\seqcard{\Gamma}+\seqcard{\Delta}+1, \seqcard{\Epsilon}+1)$ and
	$\vec n' \coloneqq (n_1, \dots, n_k, \seqcard\Gamma+1, \seqcard\Delta+1, 
	\seqcard{\Epsilon}+1)$,
\item on the left-hand side $(k+1,n_{k+1}) \cutrel (k+2,n_{k+2})$,
	\ie $F \lltens G$ is cut against $\llneg F \llpar \llneg G$,
\item on the right-hand side
    $(k+1,\seqcard{\Gamma}+1) \cutrel' (k+3,\seqcard{\Epsilon}+1)$,
    $(k+1,\seqcard{\Delta}+1) \cutrel' (k+3,\seqcard{\Epsilon}+2)$, and
	$\indi \cutrel' \indj$ if $\prempred{\indi} \cutrel \prempred{\indj}$,
	with:
	\begin{align*}
		\prempred{i,j} & \coloneqq (i,j) & \text{for $i \leq k$} \\
		\prempred{k+1,j} & \coloneqq (k+1,j) 
			& \text{for $j \leq \seqcard{\Gamma}$} \\
		\prempred{k+2,j} & \coloneqq (k+1,\seqcard{\Gamma} + j) 
			& \text{for $j \leq \seqcard{\Delta}$} \\
		\prempred{k+3,j} & \coloneqq (k+2,j) 
			& \text{for $j \leq \seqcard{\Epsilon}$} \\
		\prempred{i,j} & \coloneqq \text{undefined} & \text{otherwise}
	\end{align*}
	\ie $\cutrel'$ contains
	\begin{ienumerate}
    \item cuts between $F$ and $\llneg F$ and between $G$ and $\llneg G$,
	\item the cuts from $\cutrel$ inside $\moresequents$ or 
		between $\moresequents$ and $\Gamma$, $\Delta$ or $\Epsilon$.
	\end{ienumerate}
\end{itemize}

\medskip

\noindent
The \cutelimrule{\llwith}{\llplus} root steps are 
all the pairs of the following shape:
\cutelimstep{%
	\begin{prooftree}
	\hypo{\moresequents}
	\hypo{\vdash \Gamma, F_0}
	\hypo{\vdash \Gamma,F_1}
	\infer[double]2[\llwith]{\vdash \Gamma,F_0\llwith F_1}
	\hypo{\vdash \Delta,F_i^\bot}
	\infer[double]1[\llplus_{\mathnormal{i, F_{i-1}}}]
		{\vdash \Delta, \llneg{F_0} \llplus \llneg{F_1} } 
	\infer[double]3[mcut_{\mathnormal{k+2, \vec n, \cutrel}}]
		{\vdash \conclusion}
	\end{prooftree}
\widered
	\begin{prooftree}
	\hypo{\moresequents}
	\hypo{\vdash \Gamma,F_i}
	\hypo{\vdash \Delta,\llneg{F_i}}
	\infer[double]3[mcut_{\mathnormal{k+2, \vec n, \cutrel}}]
		{\vdash \conclusion}
	\end{prooftree}
}
such that, thanks to the additive behaviour, everything goes nicely:
\begin{itemize}
\item $\vec n \coloneqq (n_1, \dots, n_k, \seqcard\Gamma+1, 
	\seqcard\Delta+1)$,
\item $(k+1,\seqcard\Gamma+1) \cutrel (k+2,\seqcard\Delta+1)$,
	\ie on the left-hand side $F_0\llwith F_1$ is cut against $\llneg{F_0} 
	\llplus \llneg{F_1}$ and
	on the right-hand side $F_i$ is cut against $\llneg{F_i}$.
\end{itemize}

\medskip

\noindent
The \cutelimrule{\mu}{\nu} root steps are 
all the pairs of the following shape:
\cutelimstep{%
	\begin{prooftree}
	\hypo{\moresequents}
	\hypo{\vdash \Gamma, F[\mu X.F/X]}
	\infer[double]1[\mu]{\vdash \Gamma,\mu X.F}
	\hypo{\vdash \Delta, \llneg F[\nu X.\llneg F/X]}
	\infer[double]1[\nu]{\vdash  \Delta,\nu X.\llneg F} 
	\infer[double]3[mcut_{\mathnormal{k+2, \vec n, \cutrel}}]
		{\vdash \conclusion}
	\end{prooftree}
\widered
	\begin{prooftree}
	\hypo{\moresequents}
	\hypo{\vdash \Gamma, F[\mu X.F/X]}
	\hypo{\vdash \Delta, \llneg F[\nu X.\llneg F/X]}
	\infer[double]3[mcut_{\mathnormal{k+2, \vec n, \cutrel}}]
		{\vdash \conclusion}
	\end{prooftree}
}
such that:
\begin{itemize}
\item $\vec n \coloneqq (n_1, \dots, n_k, \seqcard\Gamma+1, 
	\seqcard\Delta+1)$,
\item $(k+1,\seqcard\Gamma+1) \cutrel (k+2,\seqcard\Delta+1)$,
	\ie on the left-hand side $\mu X.F$ is cut against $\nu X.\llneg F$ and
	on the right-hand side $F[\mu X.F/X]$ is cut against 
	$F[\nu X.\llneg F/X]$.
\end{itemize}

\medskip

\noindent
The \cutelimrule{\bot}{1} root steps are 
all the pairs of the following shape:
\cutelimstep{%
	\begin{prooftree}
	\hypo{\moresequents}
	\hypo{\vdash \Gamma}
	\infer[double]1[\bot]{\vdash \Gamma,\bot}
	\infer[double]0[1]{\vdash 1}
	\infer[double]3[mcut_{\mathnormal{k+2, \vec n, \cutrel}}]
		{\vdash \conclusion}
	\end{prooftree}
\widered
	\begin{prooftree}
	\hypo{\moresequents}
	\hypo{\vdash \Gamma}
	\infer[double]2[mcut_{\mathnormal{k+1, \vec n', \cutrel'}}]
		{\vdash \conclusion}
	\end{prooftree}
}
such that:
\begin{itemize}
\item on the left-hand side
	$\vec n \coloneqq (n_1, \dots, n_k, \seqcard{\Gamma}+1, 1)$
	and $(k+1,\seqcard{\Gamma}+1) \cutrel (k+2,1)$,
	\ie $\bot$ is cut against $1$,
\item on the right-hand side
	$\vec n' \coloneqq (n_1,\dots,n_k,\seqcard{\Gamma})$ and 
	$\indi \cutrel' \indj$ \ifff $\prempred{\indi} \cutrel \prempred{\indj}$,
	with:
	\begin{equation*}
		\prempred{i,j} \coloneqq (i,j) \text{ for $i \leq k$} \qquad
		\prempred{k+1,j} \coloneqq (k+1,j) 
	\end{equation*}
	\ie $\cutrel'$ contains the cuts from $\cutrel$
		inside $\moresequents$ or 
		between $\moresequents$ and $\Gamma$.
\end{itemize}

\subsection{Commutative reduction steps}

\noindent
The \cutelimrule{\llpar}{mcut} root steps are 
all the pairs of the following shape:
\cutelimstep{%
	\begin{prooftree}
	\hypo{\moresequents}
	\hypo{\vdash \Gamma,F,G}
	\infer[double]1[\llpar]{\vdash \Gamma, F\llpar G} 
	\infer[double]2[mcut_{\mathnormal{k+1,\vec n,\cutrel}}]
		{\vdash \conclusion, F\llpar G}
	\end{prooftree}
\widered
	\begin{prooftree}
	\hypo{\moresequents}
	\hypo{\vdash \Gamma,F,G}
	\infer[double]2[mcut_{\mathnormal{k+1,\vec n',\cutrel}} 
	]{\vdash \conclusion, F,G}
	\infer[double]1[\llpar]{\vdash \conclusion, F\llpar G}
	\end{prooftree}
}
such that
$\vec n \coloneqq (n_1, \dots, n_k, \seqcard{\Gamma}+1)$,
$\vec n' \coloneqq (n_1,\dots,n_k,\seqcard{\Gamma}+2)$,
and $(k+1,\seqcard{\Gamma}+1)$ is not in the support of $\cutrel$.

\medskip

\noindent
The \cutelimrule{\lltens}{mcut} root steps are 
all the pairs of the following shape:
\cutelimstep{%
	\begin{prooftree}
	\hypo{\moresequents_\Gamma}
	\hypo{\moresequents_\Delta}
	\hypo{\vdash \Gamma,F}
	\hypo{\vdash \Delta,G}
	\infer[double]2[\lltens]{\vdash \Gamma,\Delta, F\lltens G} 
	\infer[double]3[mcut_{\mathnormal{k+l+1,\vec n,\cutrel}}]
		{\vdash \conclusion_{\moresequents_\Gamma},
		\conclusion_{\moresequents_\Delta},
		\conclusion_\Gamma, \conclusion_\Delta, F\lltens G}
	\end{prooftree}
\widered
	\begin{prooftree}
	\hypo{\moresequents_\Gamma}
	\hypo{\vdash \Gamma,F}
	\infer[double]2[mcut_{\mathnormal{k+1,\vec n',\cutrel'}}]
		{\vdash \conclusion_{\moresequents_\Gamma}, \conclusion_\Gamma, F}
	\hypo{\moresequents_\Delta}
	\hypo{\vdash \Delta,G}
	\infer[double]2[mcut_{\mathnormal{l+1,\vec n'',\cutrel''}}]
		{\vdash \conclusion_{\moresequents_\Delta}, \conclusion_\Delta, G}
	\infer[double]2[\lltens]
		{\vdash \conclusion_{\moresequents_\Gamma}, \conclusion_\Gamma, 
		\conclusion_{\moresequents_\Delta}, \conclusion_\Delta, F\lltens G}
	\infer[double]1[x_{\mathnormal{\sigma}}]
		{ \vdash \conclusion_{\moresequents_\Gamma},
		\conclusion_{\moresequents_\Delta},
		\conclusion_\Gamma, \conclusion_\Delta, F\lltens G }
	\end{prooftree}
}
such that
(beware, this is the really difficult one!):
\begin{itemize}
\item $\vec n \coloneqq (n'_1, \dots, n'_k, n''_1, \dots, n''_l,
	\seqcard{\Gamma} + \seqcard{\Delta} + 1)$ is turned into
	$\vec n' \coloneqq (n'_1, \dots, n'_k, \seqcard{\Gamma} + 1)$ and
	$\vec n'' \coloneqq (n''_1, \dots, n''_l, \seqcard{\Delta} + 1)$,
\item on the left-hand side:
	\begin{itemize}
	\item $\cutrel$ relates only (indices of) formulæ
		inside $\moresequents_\Gamma$,
		inside $\moresequents_\Delta$,
		between $\moresequents_\Gamma$ and $\Gamma$, or
		between $\moresequents_\Delta$ and $\Delta$:
		such a partition of sequents is possible thanks to
		the acyclicity and connectedness assumption
		from \cref{def:multicut},
		and the corresponding reordering of the premisses of the multicut
		can be performed thanks to the rewrite rule
		introduced in \cref{app:elimination:multicut-steps},
	\item $\conclusion_{\moresequents_\Delta}$
		(resp. $\conclusion_{\moresequents_\Delta}$, $\conclusion_\Gamma$,
		$\conclusion_\Delta$) contains the formulæ from 
		$\moresequents_\Delta$
		(resp. $\moresequents_\Delta$, $\Gamma$, $\Delta$)
		whose indices are not in the support of $\cutrel$,
		and $(k+l+1,\seqcard{\Gamma} + \seqcard{\Delta} + 1)$
		is not in the support of $\cutrel$,
	\end{itemize}
\item on the right-hand side,
    $\indi \cutrel' \indj$ if $\prempred{\indi} \cutrel \prempred{\indj}$,
	with:
	\begin{align*}
		\prempred{i,j} & \coloneqq (i,j) & \text{for $i \leq k$} \\
		\prempred{k+1,j} & \coloneqq (k+l+1,j) 
			& \text{for $j \leq \seqcard{\Gamma}$} \\
		\prempred{i,j} & \coloneqq \text{undefined} & \text{otherwise}
	\intertext{and $\indi \cutrel'' \indj$ if 
	$\prempred{\indi} \cutrel \prempred{\indj}$, with:}
		\prempred{i,j} & \coloneqq (k+i,j) & \text{for $i \leq l$} \\
		\prempred{l+1,j} & \coloneqq (k+l+1, \seqcard{\Gamma} + j) 
			& \text{for $j \leq \seqcard{\Delta}$} \\
		\prempred{i,j} & \coloneqq \text{undefined} & \text{otherwise}
	\end{align*}
	\ie $\cutrel'$ and $\cutrel''$ contain the cuts from $\cutrel$
	inside $\moresequents_\Gamma$ or inside $\moresequents_\Delta$, 
	and between $\moresequents_\Gamma$ and $\Gamma$
	or $\moresequents_\Delta$ and $\Delta$;
	as for the permutation $\sigma$, it is obtained by permuting the lists
	$\conclusion_\Gamma$ and $\conclusion_{\moresequents_\Delta}$ without 
	modifying the ordering inside each of these lists
	(and leaving $\conclusion_{\moresequents_\Gamma}$ and
	$\conclusion_\Delta$ untouched).
\end{itemize}

\medskip

\noindent
The \cutelimrule{1}{mcut} root steps are 
all the pairs of the following shape:
\cutelimstep{%
	\begin{prooftree}
	\infer[double]0[1]{\vdash 1}
	\infer[double]1[mcut_{\mathnormal{1,(1),\emptyset}}]{\vdash 1}
	\end{prooftree}
\widered
	\begin{prooftree}
	\infer[double]0[1]{\vdash 1}
	\end{prooftree}
}

\medskip

\noindent
The \cutelimrule{\bot}{mcut} root steps are 
all the pairs of the following shape:
\cutelimstep{%
	\begin{prooftree}
	\hypo{\moresequents}
	\hypo{\vdash \Gamma}
	\infer[double]1[\bot]{\vdash \Gamma, \bot} 
	\infer[double]2[mcut_{\mathnormal{k+1,\vec n,\cutrel}}]
		{\vdash \conclusion, \bot}
	\end{prooftree}
\widered
	\begin{prooftree}
	\hypo{\moresequents}
	\hypo{\vdash \Gamma}
	\infer[double]2[mcut_{\mathnormal{k+1,\vec n',\cutrel}}]
		{\vdash \conclusion}
	\infer[double]1[\bot]{\vdash \conclusion, \bot}
	\end{prooftree}
}
such that
$\vec n \coloneqq (n_1, \dots, n_k, \seqcard{\Gamma}+1)$,
$\vec n' \coloneqq (n_1,\dots,n_k,\seqcard{\Gamma})$,
and $(k+1,\seqcard{\Gamma}+1)$ is not in the support of $\cutrel$.

\medskip

\noindent
Again almost nothing happens for additive constructors and fixed points:
the \cutelimrule{\llplus}{mcut}, \cutelimrule{\llwith}{mcut},
\cutelimrule{\sigma}{mcut} for $\sigma \in \{\mu,\nu\}$,
and \cutelimrule{\top}{mcut}
root steps are, respectively, all the pairs of the following shape:
\cutelimstep{%
	\begin{prooftree}
	\hypo{\moresequents}
	\hypo{\vdash \Gamma,F_i}
	\infer[double]1[\llplus_{\mathnormal{i,F_{i-1}}}]
		{\vdash \Gamma, F_0\llplus F_1} 
	\infer[double]2[mcut_{\mathnormal{k+1,\vec n,\cutrel}}]
		{\vdash \conclusion, F_0\llplus F_1}
	\end{prooftree}
\widered
	\begin{prooftree}
	\hypo{\moresequents}
	\hypo{\vdash \Gamma,F_i}
	\infer[double]2[mcut_{\mathnormal{k+1,\vec n,\cutrel}} 
	]{\vdash \conclusion, F_i}
	\infer[double]1[\llplus_{\mathnormal{i,F_{i-1}}}]
		{\vdash \conclusion, F_0\llplus F_1}
	\end{prooftree}
}

\cutelimstep{%
	\begin{prooftree}
	\hypo{\moresequents}
	\hypo{\vdash \Gamma,F}
	\hypo{\vdash \Gamma,G}
	\infer[double]2[\llwith]{\vdash \Gamma, F\llwith G} 
	\infer[double]2[mcut_{\mathnormal{k+1,\vec n,\cutrel}}]
		{\vdash \conclusion, F\llwith G}
	\end{prooftree}
\widered
	\begin{prooftree}
	\hypo{\moresequents}
	\hypo{\vdash \Gamma,F}
	\infer[double]2[mcut_{\mathnormal{k+1,\vec n',\cutrel}}]
		{\vdash \conclusion, F}
	\hypo{\moresequents}
	\hypo{\vdash \Gamma,G}
	\infer[double]2[mcut_{\mathnormal{k+1,\vec n,\cutrel}}]
		{\vdash \conclusion, G}
	\infer[double]2[\llwith]{\vdash \conclusion, F\llwith G}
	\end{prooftree}
}

\cutelimstep{%
	\begin{prooftree}
	\hypo{\moresequents}
	\hypo{\vdash \Gamma, F[\sigma X. F/X]}
	\infer[double]1[\sigma]{\vdash \Gamma, \sigma X. F} 
	\infer[double]2[mcut_{\mathnormal{k+1,\vec n,\cutrel}}]
		{\vdash \conclusion, \sigma X.F}
	\end{prooftree}
\widered
	\begin{prooftree}
	\hypo{\moresequents}
	\hypo{\vdash \Gamma,F[\sigma X. F/X]}
	\infer[double]2[mcut_{\mathnormal{k+1,\vec n,\cutrel}}]
		{\vdash \conclusion, F[\sigma X. F/X]}
	\infer[double]1[\sigma]{\vdash \conclusion, \sigma X. F}
	\end{prooftree}
}

\cutelimstep{%
	\begin{prooftree}
	\hypo{\moresequents}
	\infer[double]0[\top]{\vdash \Gamma, \top_{\!\Gamma}} 
	\infer[double]2[mcut_{\mathnormal{k+1,\vec n,\cutrel}}]
		{\vdash \conclusion, \top}
	\end{prooftree}
\widered
	\begin{prooftree}
	\infer[double]0[\top_{\!\conclusion}]{\vdash \conclusion, \top}
	\end{prooftree}
}
such that
$\vec n \coloneqq (n_1, \dots, n_k, \seqcard{\Gamma}+1)$
and $(k+1,\seqcard{\Gamma}+1)$ is not in the support of $\cutrel$.

\medskip

\noindent
The exchange \cutelimrule{x}{mcut} root steps are 
all the pairs of the following shape:
\cutelimstep{%
	\begin{prooftree}
	\hypo{\moresequents}
	\hypo{ \vdash F_{\sigma(1)}, \dots, F_{\sigma(n)} }
	\infer[double]1[x_{\mathnormal \sigma}]{\vdash F_1,\dots,F_n}
	\infer[double]2[mcut_{\mathnormal{k+1,\vec n,\cutrel}}]
		{\vdash \conclusion, F_{i_1}, \dots, F_{i_m}}
	\end{prooftree}
\widered
	\begin{prooftree}
	\hypo{\moresequents}
	\hypo{ \vdash F_{\sigma(1)}, \dots, F_{\sigma(n)} }
	\infer[double]2[mcut_{\mathnormal{k+1,\vec n,\cutrel'}}]
		{ \vdash \conclusion, F_{\sigma(i_1)}, \dots, F_{\sigma(i_n)} }
	\infer[double]1[x_{\sigma'}]{\vdash \conclusion, F_{i_1}, \dots, F_{i_m}}
	\end{prooftree}
}
such that:
\begin{itemize}
\item $\vec n \coloneqq (n_1, \dots, n_k, n)$ and
	$1 \leq i_1 < \dots < i_m \leq n$,
\item the relation $\cutrel'$ is defined by
	$\indi \cutrel' \indj$ if $\prempred{\indi} \cutrel \prempred{\indj}$,
	with:
	\[	\prempred{i,j} \coloneqq (i,j) \text{ for $i \leq k$} \qquad
		\prempred{k+1,j} \coloneqq (k+1,\sigma(j)), \]
\item the permutation $\sigma'$ acts on the set
	$[1, \seqcard{\conclusion} + m]$ and is defined by:
	\[	\sigma'(j) \coloneqq j
			\text{ for $j \leq \seqcard{\conclusion}$} \qquad
		\sigma'(\seqcard{\conclusion} + j) \coloneqq 
			\text{the $j'$ such that $\sigma(i_j) = i_{j'}$}. \]
\end{itemize}

}

\end{document}